\documentclass[reqno,9pt]{amsart}

\usepackage{amsthm, amsmath, amssymb, amsfonts, graphicx}

\usepackage[OT4]{fontenc}
\usepackage[utf8]{inputenc}
\usepackage[english]{babel}

\usepackage{bbm}

\usepackage[inline]{enumitem}

\usepackage{float}
\restylefloat{table}

\usepackage[colorlinks=true, pdfstartview=FitV, linkcolor=black,
            citecolor=black, urlcolor=black]{hyperref}
\usepackage[usenames]{color} 

\usepackage[square,numbers, sort]{natbib}

\usepackage{graphicx}
\usepackage[font=sl,labelfont=bf,width=\textwidth]{caption}
\usepackage{color}

\DeclareMathOperator{\erf}{erf}

\DeclareMathOperator{\cov}{Cov}

\newcommand{\stab}{{$\alpha$-stable }}

\usepackage{multirow}
\usepackage{amsthm}
\newtheorem{theorem}{Theorem}

\newtheorem{proposition}[theorem]{Proposition}

\def\bP{\mathbb{P}}

\def\bR{\mathbb{R}}
\def\bN{\mathbb{N}}
\def\cF{\mathcal{F}}
\usepackage{bbm}
\newcommand{\1}{\mathbbm{1}}

\usepackage{amssymb}

\usepackage{placeins} 
\usepackage{graphicx}

\usepackage{setspace}
\usepackage[foot]{amsaddr}

\begin{document}

\title[Goodness-of-fit tests for the one-sided L\'evy distribution based on quantile conditional moments]{Goodness-of-fit tests for the one-sided L\'evy distribution based on quantile conditional moments}

\author{Kewin P\k{a}czek$^{\ast,\S}$}
\author{Damian Jelito$^{\ast}$}
\author{Marcin Pitera$^{\ast}$}
\author{Agnieszka Wy\l{}oma\'{n}ska$^{\dagger}$}
\address{$^{\ast}$Institute of Mathematics, Jagiellonian University, S. {\L}ojasiewicza 6, 30-348 Krak{\'o}w, Poland}
\address{$^{\dagger}$Faculty of Pure and Applied Mathematics, Hugo Steinhaus Center, Wroc{\l}aw University of Science and Technology, Wyspia{\'n}skiego 27, 50-370 Wroc{\l}aw, Poland}
\address{$^{\S}$Corresponding author}
\email{kewin.paczek@im.uj.edu.pl, damian.jelito@uj.edu.pl, marcin.pitera@uj.edu.pl, agnieszka.wylomanska@pwr.edu.pl}

\maketitle

\vspace{-1cm}
\begin{abstract}
In this paper we introduce a novel statistical framework based on the first two quantile conditional moments that facilitates effective goodness-of-fit testing for one-sided L\'evy distributions. The scale-ratio framework introduced in this paper extends our previous results in which we have shown how to extract unique distribution features using conditional variance  ratio for the generic class of \stab distributions. We show that the conditional moment-based goodness-of-fit statistics are a good alternative to other methods  introduced in the literature tailored to the one-sided L\'evy distributions.  The usefulness of our approach is verified using an empirical test power study. For completeness, we also derive the asymptotic distributions of the test statistics and show how to apply our framework to real data. 

\vspace{0.2cm}

\end{abstract}

\section{Introduction}



The one-sided L\'evy distribution is a subclass of \stab distributions that were introduced in \cite{levy1924theorie}, see \cite{Nolan2020} for a recent survey on this class of distributions. This subclass is a popular modelling choice in many disciplines including physics and natural sciences, see \cite{Tsa1997,Tsa1995,Tsa1999}, evolutionary programming, see \cite{Yao2004}, photonics, see \cite{Gom2016}, and image processing, see \cite{Ran2015}, see also \cite{Rog2008} and references therein. The popularity of the one-sided L\'evy distribution as a modelling choice could be explained by the fact that -- together with Cauchy and Gaussian distribution -- it is the only \stab distribution for which the probability density function  and the cumulative distribution function could be explicitly stated in analytical form, see \cite{stable}, \cite{Ahs2019}, and \cite{Rue1998} for details. 

Despite its numerous interesting applications, the usage of one-sided L\'evy distribution presents some modelling challenges. The most crucial one is the lack of the finite expected value and variance due to the heavy right tail. Because of that, the classical approaches, e.g. the method of moments, cannot be directly applied and special treatment is required. In the context of goodness-of-fit testing, a few methods tailored specifically to one-sided L\'evy distributions have been developed. 
Recent contributions includes a {\it jackknife empirical likelihood ratio} method, {\it Stein’s Characterization based} method, {\it scale ratio estimation} method, and {\it Laplace transform-based approach}. For details, we refer to \cite{BhaKat2020}, \cite{kumari2023}, \cite{KumBha2022}, \cite{Lukic2023}, and references therein.

Most of the  one-sided L\'evy distribution goodness-of-fit methods proposed in the literature overcome the problem of infinite moments by introducing special sample transforms. In this paper, we introduce an alternative, more direct framework, based on quantile sample conditioning. Namely, we show how to construct effective goodness-of-fit test statistics based on quantile conditional mean ratio or quantile conditional variance ratio. Our approach could be also used to refine the existing frameworks using ensemble methods.

The statistical framework introduced in this paper could be seen as a tailored extension of the quantile conditional variance method introduced in~\cite{PacJelPitWyl2022}, where a fitting framework for a generic class of \stab distributions was considered. The method introduced in this paper also complements the analysis presented in \cite{PitCheWyl2021}, where it has been shown that the algorithms based on quantile conditional variance statistic could be used for effective estimation and testing of the general class of \stab distributions, and in \cite{JelPit2018}, where a  specific goodness-of-fit test statistic that measures tail heaviness in reference to Gaussian distribution has been introduced. In a more general context, the quantile conditional moments-based methods have been recently studied e.g. in the context of independence characterisation or to explain empirical phenomena such as the 20/60/20 rule, see \cite{JawJelPit2022} and \cite{JawPit2015}.
The conditional moment-based  methodologies have been also recently applied to various practical problems including local damage detection based on the signals with heavy-tailed background noise or to study the asymptotic behaviour of empirical processes, see \cite{HebZimWyl2020,HebZimPitWyl2019,Ghoudi2018}.

In the present paper, the quantile conditional moment approach is proposed for the estimation of the scale parameters of the one-sided L\'evy distribution and later used to construct ratio-based statistic that refines known algorithms for the goodness-of-fit testing. Although the theoretical moments (including expected value) for the one-sided L\'evy distribution are infinite, the corresponding conditional moments always exist and can be used to characterize this distribution, see  \cite{JawPit2020}. Thus, the proposed methodology in a natural way can be considered as the generalisation of the method of moments applied to infinite-variance distributions. We demonstrate that the proposed approach improves the effectiveness of the goodness-of-fit testing when applying ratio statistics. The effectiveness is tested for plenty of known distributions that are considered as difficult to distinguish from one-sided L\'evy distribution. The results are compared to the approach recently introduced in \cite{KumBha2022} using both simulated and real data samples. The new methodology based on the quantile conditional moments seems to be more effective for small sample sizes which is the most challenging case in real applications.   Finally,  the new approach is intuitive and easy to implement, which is a crucial aspect of real data analysis. It is worth noting that in \cite{Lukic2023} the conditional variance ratio  has been proposed for the one-sided L\'evy distribution goodness-of-fit testing. That said, the introduced statistic was indirectly assuming symmetry and did not incorporate information about the pre-fixed location, which substantially reduced the statistical power of this test. This was due to the fact that the authors took the test statistic introduced \cite{PitCheWyl2021} in which the authors considered symmetric distributions in the scale-location invariant setting; the corresponding test statistic was tailored to this situation rather than to highly asymmetric one-sided L\'evy setting in which the location is fixed. Consequently, our paper could be also seen as a refinement of the method introduced in \cite{PitCheWyl2021} that accounts for asymmetry and different location processing.

This paper is organised as follows. In Section~\ref{S:preliminaries}  we introduce the one-sided L\'{e}vy distribution, set up the notation, and provide analytical formulas for quantile conditional mean and variance. We also introduce the sample quantile moments and quantile conditional moments-based statistics used for scale parameter estimation. In Section \ref{sec3} we discuss statistical goodness-of-fit tests for one-sided L\'{e}vy distribution and propose the ratio test statistics based on quantile conditional moments. Moreover, we discuss the probabilistic properties of the introduced test statistics. Next, in Section \ref{sec4} we demonstrate the effectiveness of the quantile conditional moments-based approach for testing the one-sided L\'evy distribution. Here we discuss the power of the introduced tests for numerous alternative distributions and various sample sizes. In Section \ref{sec5} we introduce the goodness-of-fit test with the location-invariant test statistic. We demonstrate the simulation study to show the effectiveness of the proposed approach. Section \ref{sec6} contains the real data application. The last section concludes the paper.

\section{Preliminaries}\label{S:preliminaries}
Let $(\Omega,\cF,\bP)$ be a probability space. We say that a random variable $X$ follows the {\it one-sided L\'{e}vy distribution} with scale parameter $c>0$, and write $X\sim Lv(c)$, if its cumulative distribution function (CDF) is given by
\begin{equation}\label{eq:Levycdf}
F(x;c) := 
\begin{cases}
0, & x\leq 0,\\
2 - 2\Phi\left(\sqrt{\frac{c}{x}}\right),&  x> 0,
\end{cases}
\end{equation}
where $\Phi(\cdot)$ is the standard normal CDF. In this case, the probability density function (PDF) and quantile function for $X\sim Lv(c)$ are given by
\begin{align}
 f(x;c) & := \1_{[0,\infty)}(x)\cdot \sqrt{\frac{c}{2\pi}}x^{-3/2}\exp\left(-\frac{c}{2x}\right), \quad x\in \bR,\label{eq:Levydens}\\
 Q(p;c) &:= \frac{c}{\left(\Phi^{-1}(1-p/2)\right)^2} =  \frac{c}{\left(\sqrt{2} G(p)\right)^2}, \quad p\in (0,1),\label{eq:Levy_Q}
\end{align}
respectively, where $G(p) := \erf^{-1}(1-p)$, for $p\in (0,1)$, and $\erf$ is the standard error function.
It should be noted that the one-sided L\'{e}vy distribution belongs to the class of \stab  distributions with stability index $\alpha = 1/2$ and skewness parameter $\beta = 1$. Thus, in general, we might also consider an extended class of one-sided L\'{e}vy distributions with additional location parameter $\mu\in\bR$. In this case, which is discussed mainly in Section~\ref{sec5}, we write $X\sim  Lv(\mu,c)$ and adjust all presented formulas to account for additional location shift, see e.g. Chapter 1 in~\cite{Nolan2020}. Also, it is worth noting that $X\sim Lv(c)$ could be seen as the inverse Gamma distribution with the shape parameter equal to $1/2$ and the scale parameter equal to $c/2$; see \cite{CragHack}. Consequently, while the moments of $X\sim Lv(c)$ are infinite, the inverse moments, i.e. moments of $X^{-1}$, are finite. 

Let us now introduce a notation related to the quantile conditional moments of a generic (absolutely continuous) random variable $X$. 
For a quantile split $0\leq a<b\leq 1$, we define the quantile conditioning set of $X$ by setting
\[
M_X(a,b):=\{X\in[ F_X^{-1}(a),F_X^{-1}(b)]\},
\]
where $F_X^{-1}(\cdot)$ is the (generalised) inverse of the CDF of $X$. For quantile splits $0\leq a<b\leq 1$, the related quantile conditional mean (QCM) and quantile conditional variance (QCV) of $X$ are given by
\begin{align}
\mu_X(a,b) &:= \mathbb{E}[X| M_X(a,b) ],\label{qe:con_mean}\\
\sigma_X^2(a,b) &:=\mathbb{E}[X^2| M_X(a,b)]-(\mathbb{E}[X| M_X(a,b)])^2,\label{eq:cond_var}
 \end{align}
respectively. Note that for any $0<a<b<1$ and generic random variable $X$ both $\mu_X(a,b)$ and $\sigma^2_X(a,b)$ are well-defined and finite since $X$ is bounded on $M_X(a,b)$ and $\bP[M_X(a,b)]>0$. For $X \sim Lv(c)$, one can compute QCM and QCV explicitly. Indeed, assuming that $X \sim Lv(c)$, for any $0\leq a<b<1$, we get
\begin{align}
\mu_{X}(a,b) &= \frac{\sqrt{\pi} c \left( \exp\left(-G^2(b)\right)G(a) - \exp \left(-G^2(a)\right)G(b)\right)}{(b-a) \pi  G(a)G(b)} -c, \label{eq:cond_mean_quant}\\
\sigma^2_{X}(a,b)&=  \frac{-c^2 \left( \exp\left(-G^2(b)\right)\left(G^2(a) - \frac{1}{2}\right) + \exp\left(-G^2(a)\right)\left(G^2(b) - \frac{1}{2}\right) \right) }{3\sqrt(\pi)(b-a)G^3(b)G^3(a)}  - \frac{c^2}{3(b-a)}
-\mu_{X}(a,b)^2.\label{eq:cond_var_quant}
\end{align}

Both QCM and QCV can be estimated using sample estimators similar to the ones used for unconditional moments. Namely, for $n\in\bN$,  a sample of independent and identically distributed (i.i.d) random variables $(X_i)_{i=1}^{n}$, and a quantile split $0 \leq a < b \leq 1$, the {\it sample quantile conditional mean}  and {\it sample quantile conditional variance} are given by
\begin{align}
     \hat{\mu}_X(a,b) &:= \textstyle \frac{1}{[nb] - [na]} \sum_{i=[na]+1}^{[nb]}X_{(i)},\label{eq:estimQCM}\\
 \hat{\sigma}^2_X(a,b) &:= \textstyle \frac{1}{[nb] - [na]} \sum_{i=[na]+1}^{[nb]} (X_{(i)} - \hat{\mu}_X(a,b))^2,\label{eq:estimQCV}
\end{align}
respectively, where $X_{(i)}$ corresponds to the $i$th order statistic of the sample, and  $[x] := \max\{k \in \mathbb{Z}: k \leq x \}$ is the integer part of $x\in\bR$. Estimators \eqref{eq:estimQCM} and \eqref{eq:estimQCV} could be seen as a classical sample mean and sample variance applied to the conditioned sub-sample in which the conditioning is based on order statistics. From \cite{JelPit2018} and \cite{PitCheWyl2021} we know that estimators \eqref{eq:estimQCM} and \eqref{eq:estimQCV} are consistent and follow CLT-type law, i.e. they have asymptotically normal distributions for i.i.d samples. 



Now, we show how to use~\eqref{eq:estimQCM} and \eqref{eq:estimQCV} to estimate the scale parameter for one-sided L\'{e}vy distribution. Let us assume we are given an arbitrarily chosen quantile split $0\leq a<b<1$ and an i.i.d. sample $(X_i)_{i=1}^{n}$ from $Lv(c)$ with unknown scale parameter $c>0$. Noting that QCM and QCV are finite and proportional to $c$, we can estimate the scale parameter simply by setting
\begin{align}
    \hat{c}_{\textrm{QCM}}(a,b)&:=\frac{\hat{\mu}_{X}(a,b)}{\mu_{Lv(1)}(a,b)},\\
    \hat{c}_{\textrm{QCV}}(a,b)&:=\sqrt{\frac{\hat{\sigma}^2_{X}(a,b)}{\sigma^2_{Lv(1)}(a,b)}},
\end{align}
where $\mu_{Lv(1)}(a,b)$ and $\sigma^2_{Lv(1)}(a,b)$ denote the normalising constants equal to $\mu_W(a,b)$ and $\sigma^2_W(a,b)$ for $W\sim Lv(1)$; note that these constants can be computed using~\eqref{eq:cond_mean_quant} and~\eqref{eq:cond_var_quant} with $c=1$.

In the literature, there are other methods designed to estimate scale parameter $c>0$. Let us only mention the classical maximum likelihood estimator and a more recent idea related to the covariance of suitable chosen random variables proposed in~\cite{KumBha2022}. More specifically, given an i.i.d. sample $(X_i)_{i=1}^{n}$ from $X\sim Lv(c)$, we define the transformed samples  $(Y_i)_{i=1}^{n}$ and $(Z_i)_{i=1}^{n}$, where $Y_i:=1/X_i$ and $Z_i:=\ln Y_i$, for $i=1,\ldots,n$. Then, following~\cite{KumBha2022}, we define the maximum likelihood (MLE) scale estimator and  covariance-based (COV) scale estimator by setting
\begin{align} 
    \hat{c}_{\textrm{MLE}}&:=\frac{1}{\overline{Y}_n},\\
    \hat{c}_{\textrm{COV}}&:=\frac{2n}{\sum_{i=1}^n (Y_i-\overline{Y}_n)(Z_i-\overline{Z}_n)},
\end{align} 
where  $\bar{Y}_n$ and $\bar{Z}_n$ are empirical sample means of $(Y_i)_{i=1}^{n}$ and $(Z_i)_{i=1}^{n}$, respectively. It should be noted that the estimators  $\hat{c}_{\textrm{MLE}}$ and $\hat{c}_{\textrm{COV}}$ are consistent and one can derive their asymptotic distributions, see~\cite{KumBha2022} for details. 


The high-level performance comparison of the introduced estimators is presented in Figure~\ref{fig:boxplot}. It should be noted that the MLE estimator outperforms all other estimators, and the second best estimator is the QCM estimator. In general, the COV estimator presents some small bias, while the QCV estimator has the highest standard error. It should be noted that the worse performance of the QCV estimator could be traced back to the fact that it is designed to be location-invariant while the other estimators are not.


\begin{figure}[htp!]
\begin{center}
\includegraphics[width=0.22\textwidth]{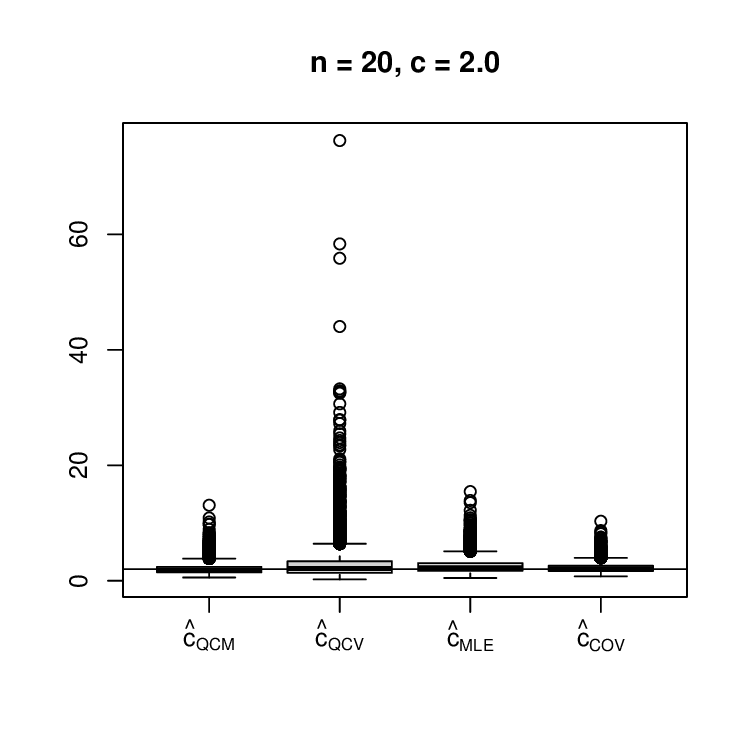}
\includegraphics[width=0.22\textwidth]{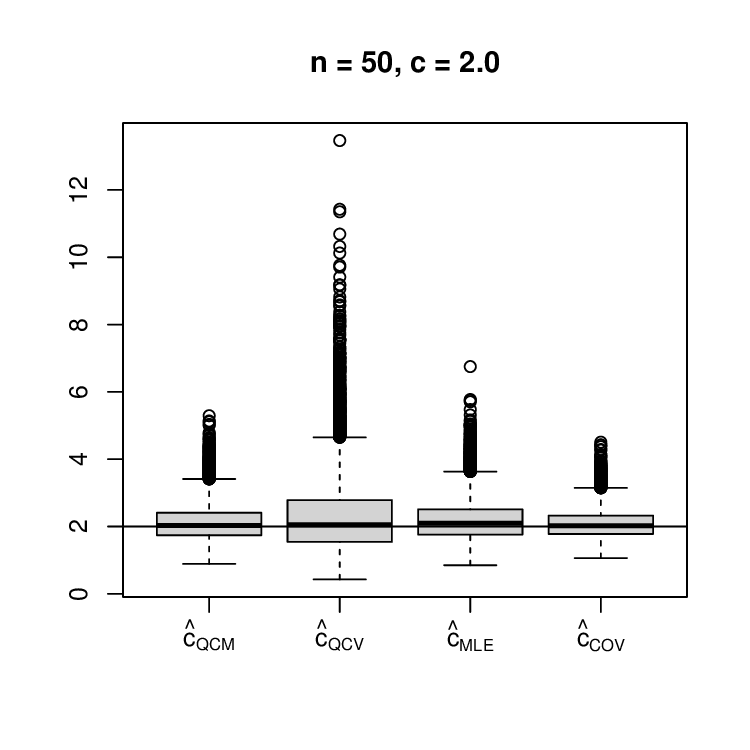}
\includegraphics[width=0.22\textwidth]{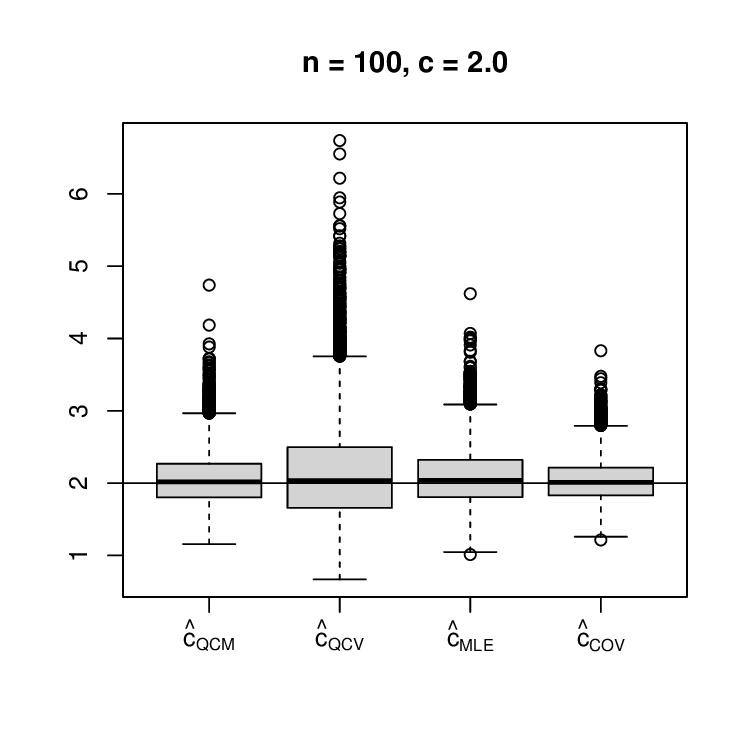}
\includegraphics[width=0.22\textwidth]{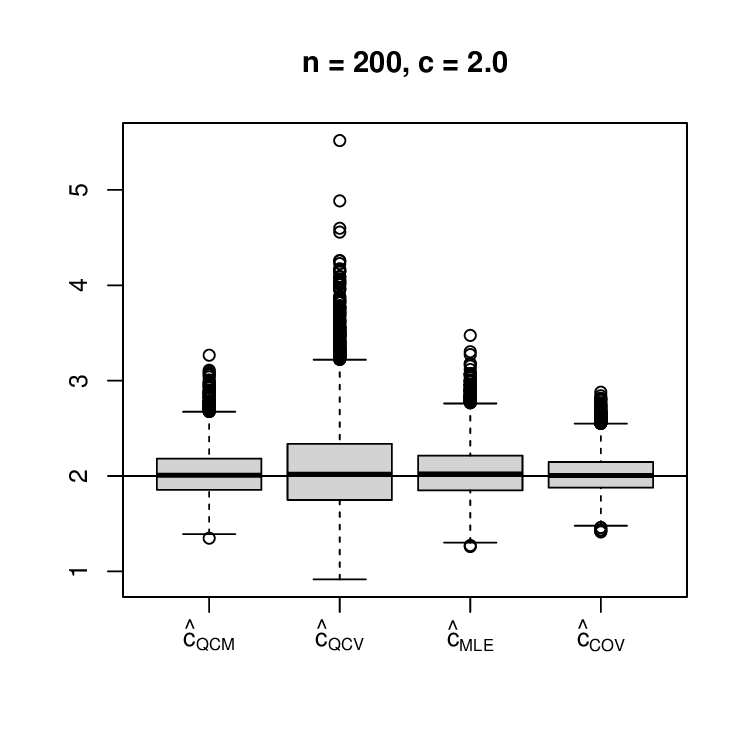}
\end{center}
\caption{ Boxplots of the estimated values for the quantile conditional mean estimator ($\hat{c}_{\textrm{QCM}}$) with $a=0.2$ and $b=0.48$, the quantile conditional variance ($\hat{c}_{\textrm{QCV}}$) with $a = 0.0$ and $b = 0.7$, the maximum likelihood estimator ($\hat{c}_{\textrm{MLE}}$), and the method of covariance estimator ($\hat{c}_{\textrm{COV}}$). The results are based on 10000 strong Monte Carlo samples with sizes $n \in \{20,50,100,200\}$. The horizontal black line represents the true value of the scale parameter $c = 2$. }
\label{fig:boxplot}
\end{figure}

\section{One-sided L\'{e}vy distribution goodness-of-fit test statistics}\label{sec3}

In this section we are interested in statistical goodness-of-fit tests for the null hypothesis stating that the sample follows the one-sided L\'{e}vy distribution (with any $c>0$) and an alternative stating that the sample is from some other family of distributions.  As in the previous section, we follow the standard i.i.d. statistical framework, use $X$ to denote the reference random variable and $(X_i)_{i=1}^{n}$, $n\in\bN$, to denote sample from $X$. Also, as in the previous section, we use $(Y_i)_{i=1}^n$ and $(Z_i)_{i=1}^n$ to denote the transformed samples, where $Y_i=1/X_i$ and $Z_i=\ln Y_i$, for $i=1,\ldots,n$.

The aim of this section is to show how to refine the existing testing frameworks using test statistics based on quantile conditional moments. For the key benchmark, we take the test statistic introduced in~\cite{KumBha2022} which is based on the ratio between two scale estimators, i.e. $\hat{c}_{\textrm{COV}}$ and $\hat{c}_{\textrm{MLE}}$. Following this research, we introduce the test statistic
\begin{equation}\label{eq:Vn}
   V_n:= \sqrt{n}\left(\frac{ \hat c_{\textrm{COV}}}{ \hat c_{\textrm{MLE}}} -1\right).
\end{equation}
It should be noted that $V_n$ is a pivotal quantity with respect to $c>0$ and the value of $V_n$ should be close to $0$ if the sample follows the one-sided L\'{e}vy distribution. For completeness, we restate the result from~\cite{KumBha2022}, where the asymptotic distribution of $V_n$ is derived; see Theorem 3.1 and Corollary 3.1 therein for details.

\begin{proposition}\label{pr:benchmark}
Let $X_i$, $i = 1,2, \ldots$, be i.i.d. random variables following $Lv(c)$ with some $c>0$. Also, let $V_n$ be given by~\eqref{eq:Vn}. Then, we get $V_n\overset{d}{\to }\mathcal{N}(0,\tau^2_V)$, as $n \to \infty$, where $\tau_V > 0$ is a constant independent of $c$. 
\end{proposition}
The proof of Proposition~\ref{pr:benchmark} could be found in~\cite{KumBha2022}. Proposition~\ref{pr:benchmark} is also a direct implication of Proposition~\ref{pr:appendix}, which proof we defer to the appendix. The graphical illustration of Proposition~\ref{pr:benchmark} is presented in Figure~\ref{fig:Vhat}.

\begin{figure}[htp!]
\begin{center}
\includegraphics[width=0.26\textwidth]{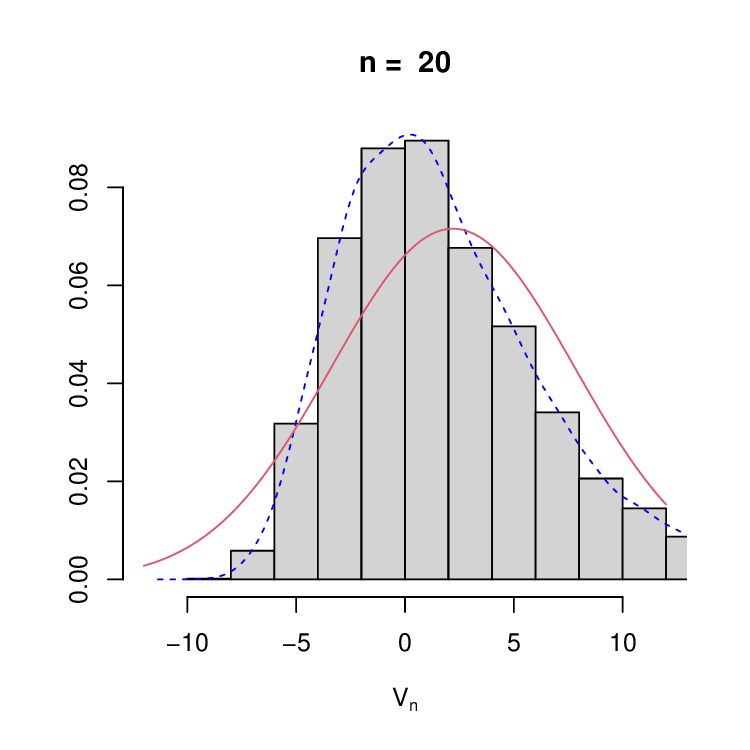}
\includegraphics[width=0.26\textwidth]{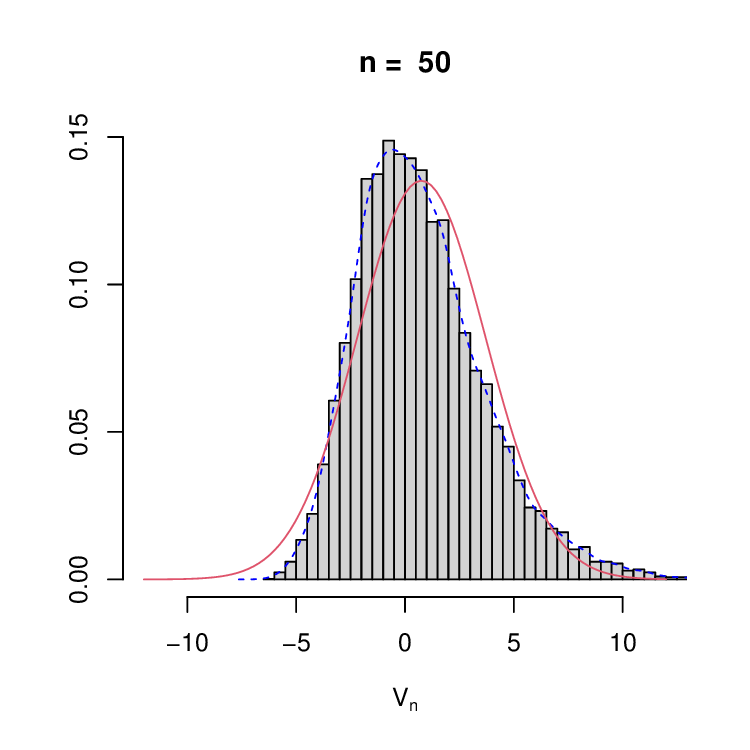}
\includegraphics[width=0.26\textwidth]{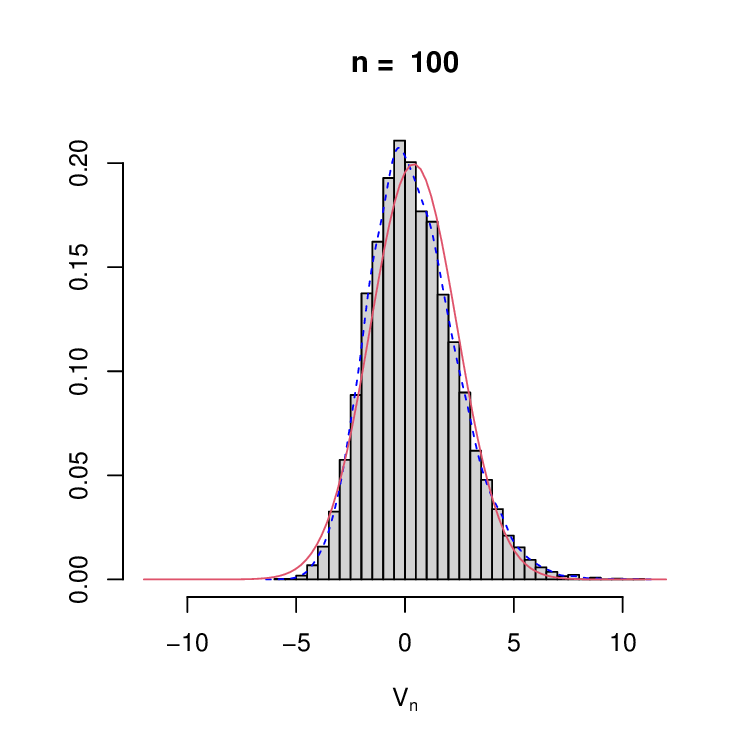}
\includegraphics[width=0.26\textwidth]{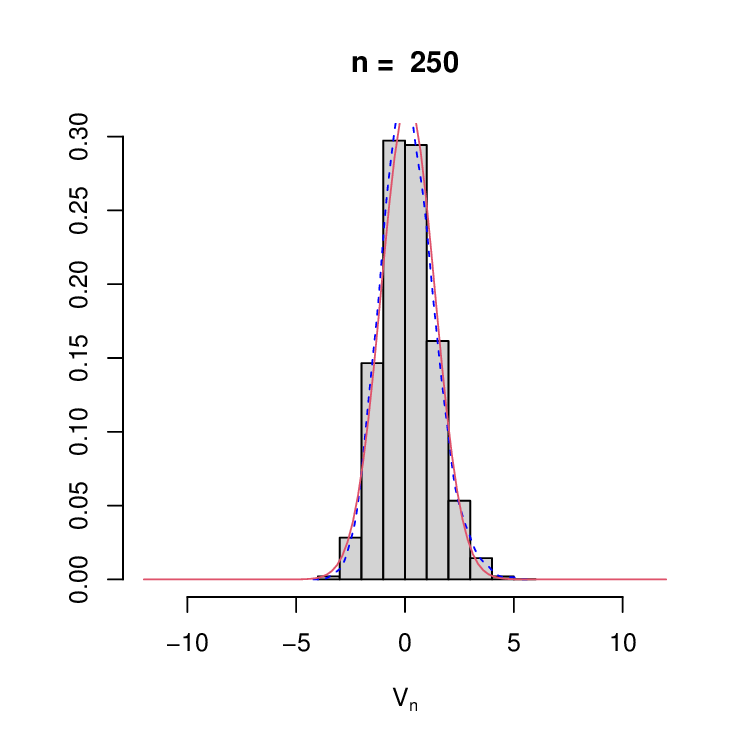}
\includegraphics[width=0.26\textwidth]{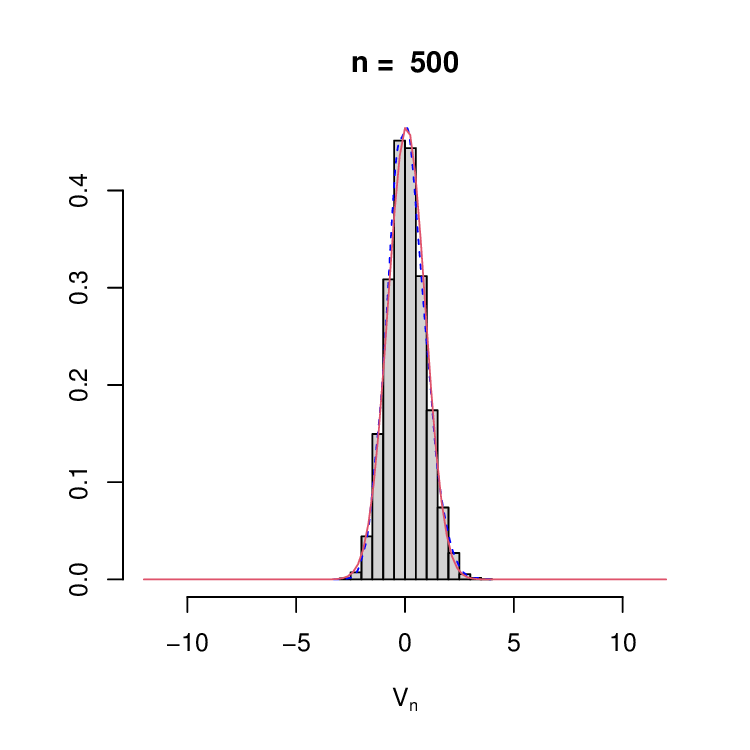}
\includegraphics[width=0.26\textwidth]{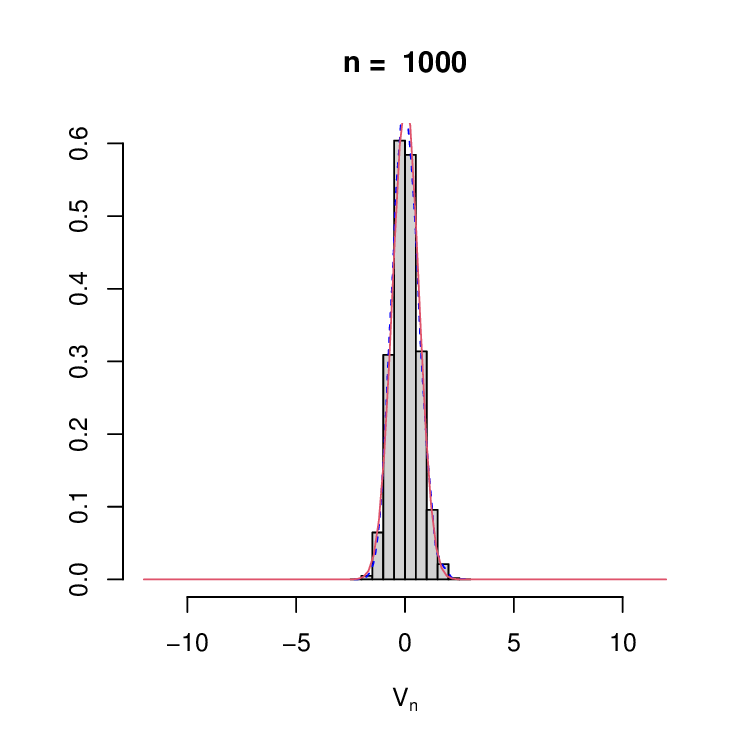}
\end{center}
\caption{Histogram of $V_n$ distribution under the null hypothesis with $n \in \{20,50,100,250,500,1000\}$. The blue dashed line represents the kernel density while the solid red line shows the fitted normal density. Each histogram is based on 10 000 strong Monte Carlo simulations.}
\label{fig:Vhat}
\end{figure} 
In the following subsections, we propose two alternative new test statistics that are based on quantile conditional moments.

\subsection{Scale ratio statistic based on QCM}
Following the logic used in \eqref{eq:Vn}, we can define a test statistic which uses scale estimators based on QCM characteristics for two different quantile splits. Namely, we define a new test statistic by setting
\begin{equation}\label{eq:On}
    O_n := \sqrt{n} \left(\frac{\hat{c}_{\textrm{QCM}}(a_1,b_1)}{\hat{c}_{\textrm{QCM}}(a_2,b_2)} -1\right),
\end{equation}
where $0\leq a_1<b_1<1$ and $0\leq a_2<b_2< 1$ are two fixed quantile splits. Note that  $O_n$ is a pivotal quantity with respect to scaling. For testing purposes we consider fixed values $a_1=0$, $b_1=0.3$, $a_2=0.8$, and $b_2=0.95$.
As in Proposition~\ref{pr:benchmark}, we may show that $O_n$ is asymptotically normal under the null hypothesis.

\begin{proposition}\label{pr:On.stat}
Let $X_i$, $i = 1,2, \ldots$, be i.i.d. random variables following $Lv(c)$ with some $c>0$. Also, let $O_n$ be given by~\eqref{eq:On}. Then, we get
$O_n\overset{d}{\to} \mathcal{N}(0, \tau_O^2)$, as $n\to\infty$, where $\tau_O>0$ is a constant independent of $c$.
\end{proposition}

Proposition~\ref{pr:On.stat} is a direct implication of Proposition~\ref{pr:appendix}, which proof is deferred to the Appendix. The graphical illustration of Proposition~\ref{pr:On.stat} could be found in Figure~\ref{fig:Ohat}.

  \begin{figure}[htp!]
\begin{center}
\includegraphics[width=0.26\textwidth]{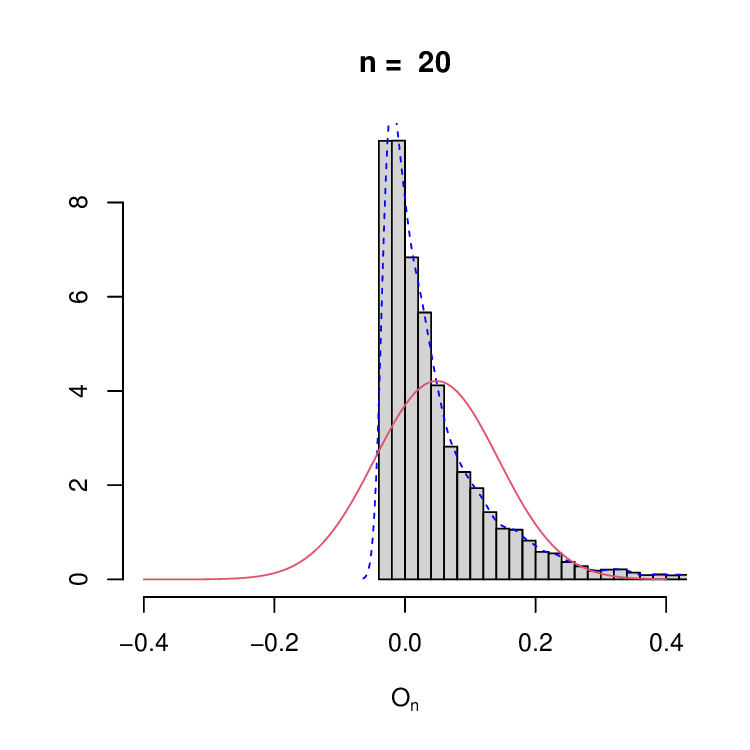}
\includegraphics[width=0.26\textwidth]{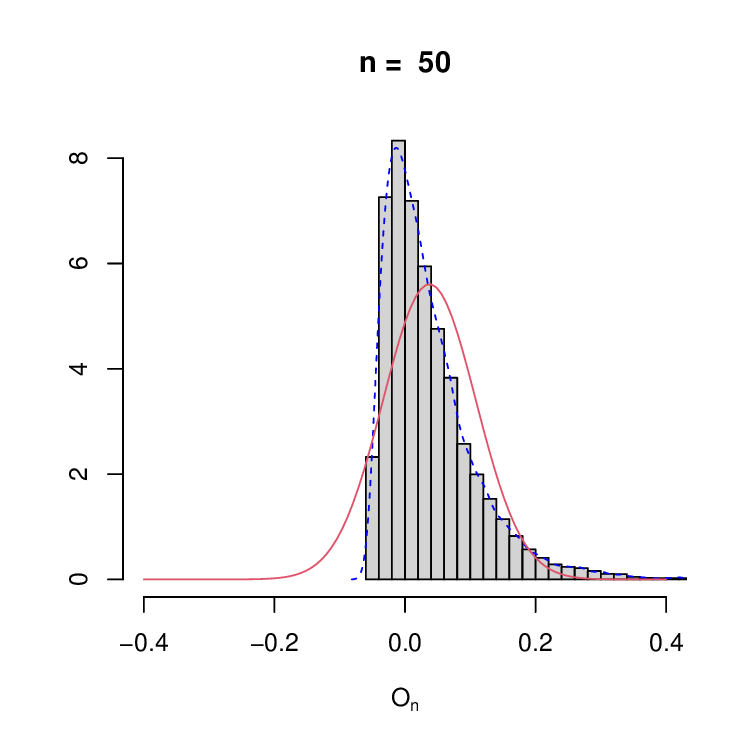}
\includegraphics[width=0.26\textwidth]{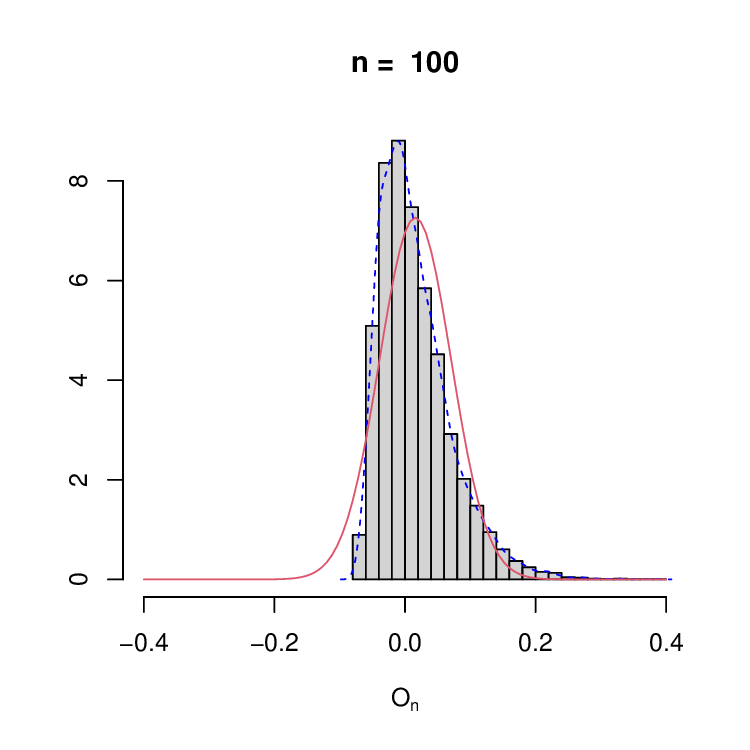}
\includegraphics[width=0.26\textwidth]{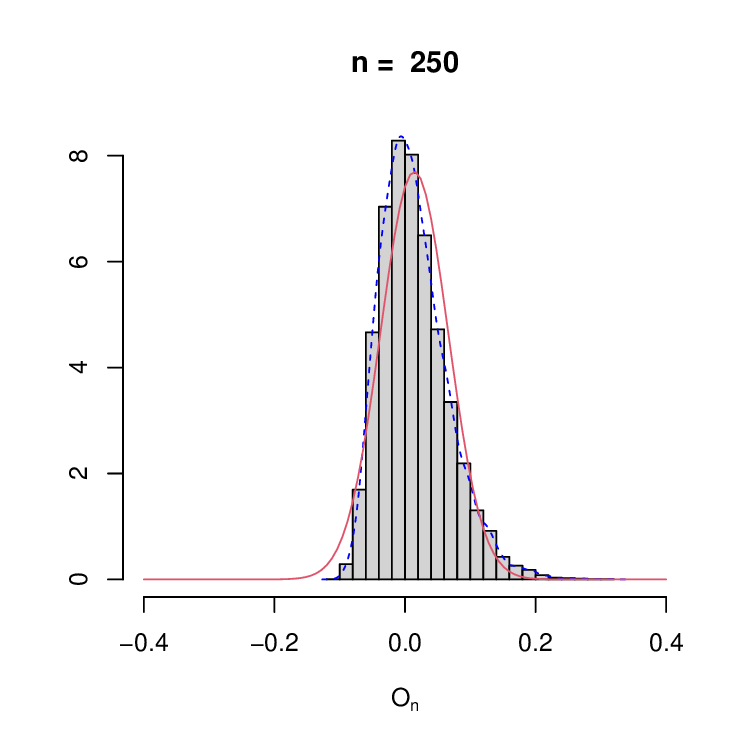}
\includegraphics[width=0.26\textwidth]{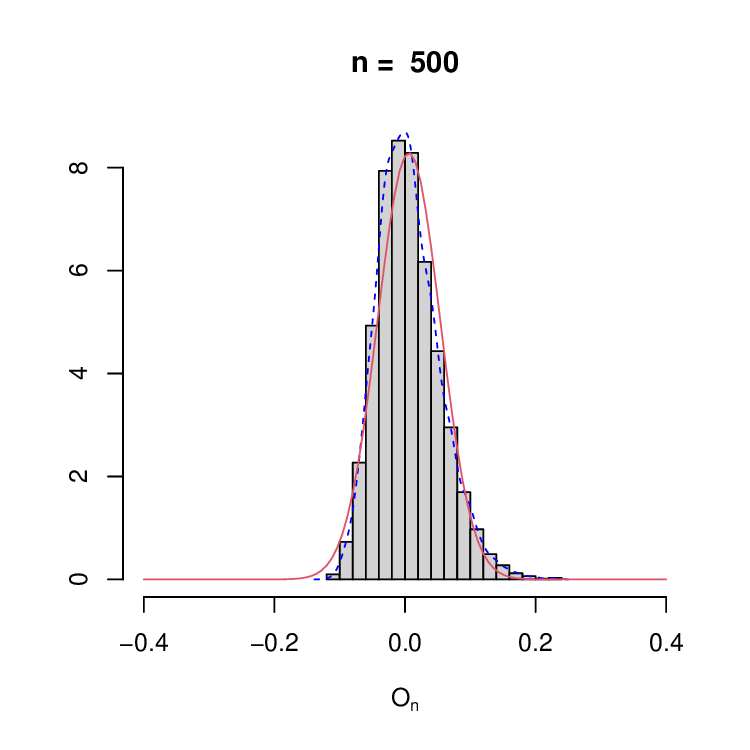}
\includegraphics[width=0.26\textwidth]{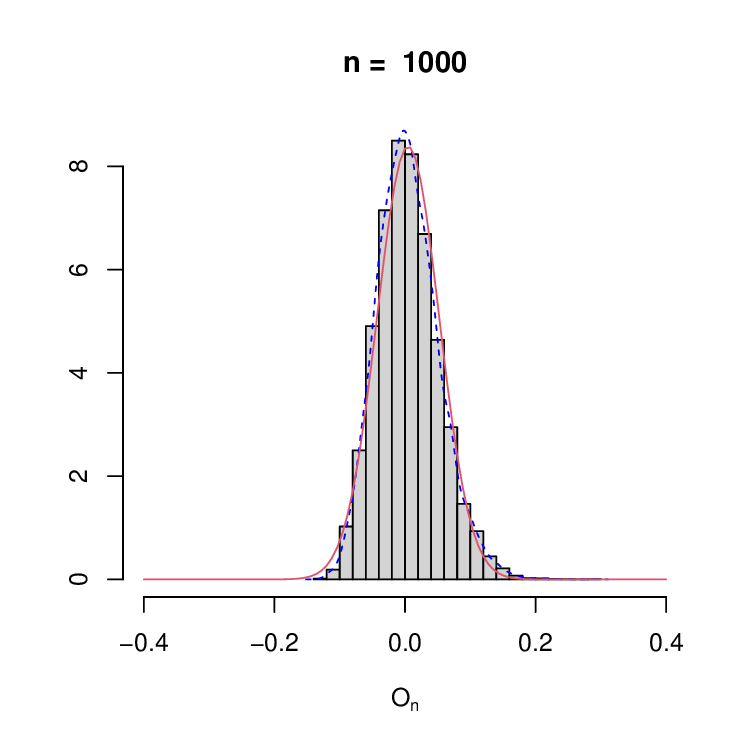}
\end{center}
\caption{Histogram of $O_n$ distribution under the null hypothesis with $n \in \{20,50,100,250,500,1000\}$. The blue dashed line represents the kernel density while the solid red line shows the fitted normal density. Each histogram is based on 10 000 strong Monte Carlo simulations.}
\label{fig:Ohat}
\end{figure}

\FloatBarrier

\subsection{Ensemble scale ratio test statistic}
The next test statistic is also inspired by~\eqref{eq:Vn} and uses $\hat c_{\textrm{COV}}$ estimator. More specifically, we consider an ensemble statistic that modifies $V_n$ introduced in~\eqref{eq:Vn} and is equal to \begin{align}\label{eq:Tn}
T_n:=\sqrt{n}\left(\frac{ \hat c_{\textrm{COV}}+\hat c_{\textrm{QCM}}(a,b)}{2\cdot \hat c_{\textrm{MLE}}} -1\right),
\end{align}
where $0\leq a<b<1$ is a fixed quantile split. For testing purposes, we consider fixed values $a=0.02$ and $b=0.48$ that
were chosen empirically to minimize the standard deviation of the absolute estimation error of $\hat{c}_{\textrm{QCM}}(a,b)$.

Before we proceed, let us explain the rationale behind the construction of $T_n$. First, in Figure~\ref{fig:Ccorr}, we confront the estimated values of $c$ for $Lv(c)$ distribution, for some estimators discussed in Section~\ref{S:preliminaries}. From this plot, we can see that the correlation between the ensemble estimator $\tfrac{1}{2}(\hat c_{\textrm{COV}}+\hat c_{\textrm{QCM}}(a,b))$ and $\hat{c}_{MLE}$ is higher than the correlation between estimators $\hat{c}_{COV}$ and $\hat{c}_{MLE}$ which should have a positive impact on the power of the test, e.g. due to the fact that the ratio between those statistics might be more focused around 1. Also, note that $\hat{c}_{COV}$ and $\hat{c}_{QCM}$ are not highly correlated so the combined sum could possess more statistical information within. This is numerically summarised in Table~\ref{tab:correlations}, where we present the covariances and the correlations for the considered estimators.

\begin{figure}[htp!]
\begin{center}
\includegraphics[width=0.22\textwidth]{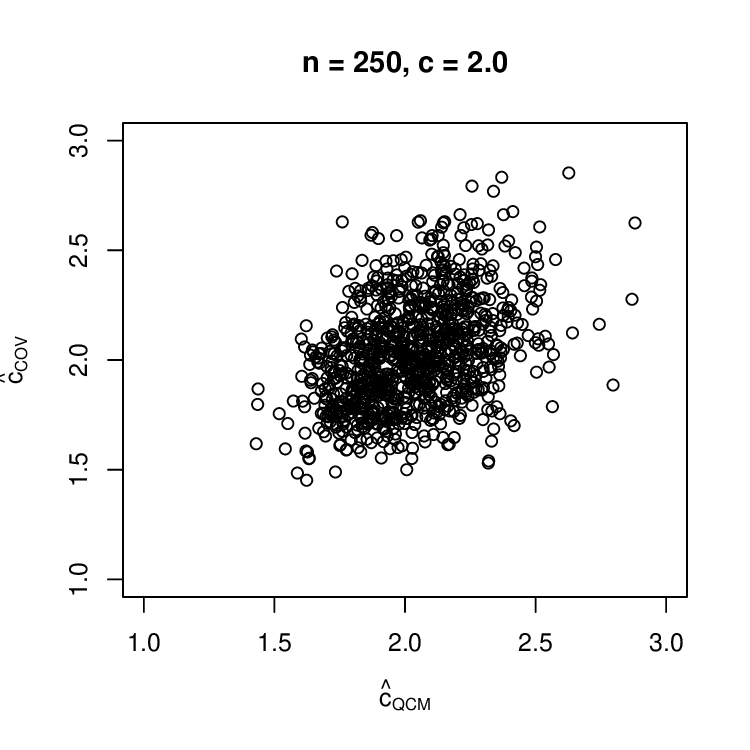}
\includegraphics[width=0.22\textwidth]{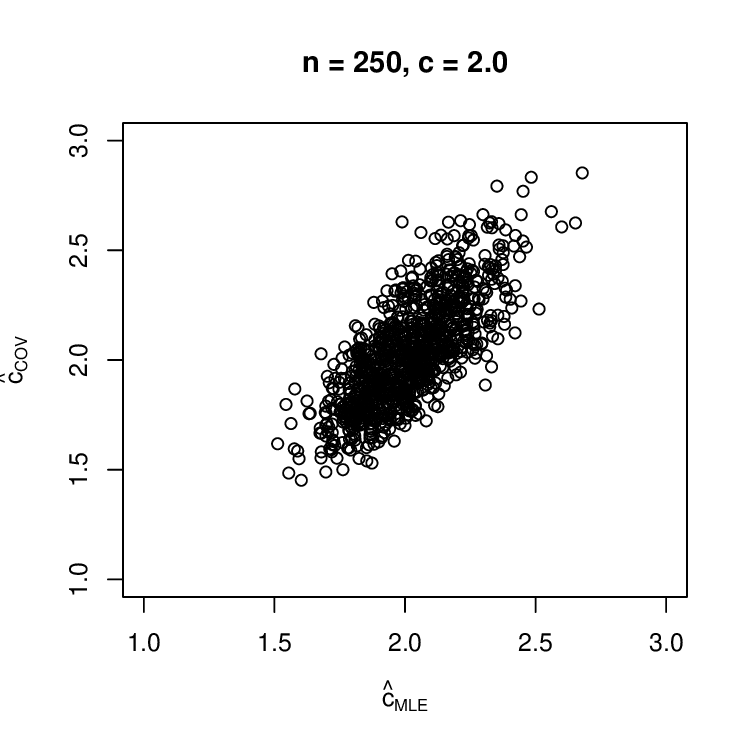}
\includegraphics[width=0.22\textwidth]{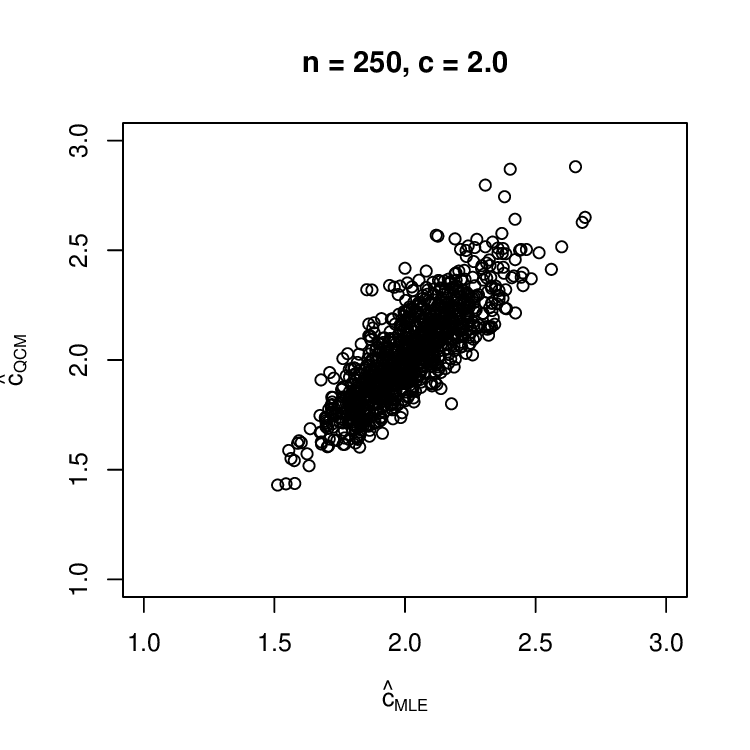}
\includegraphics[width=0.22\textwidth]{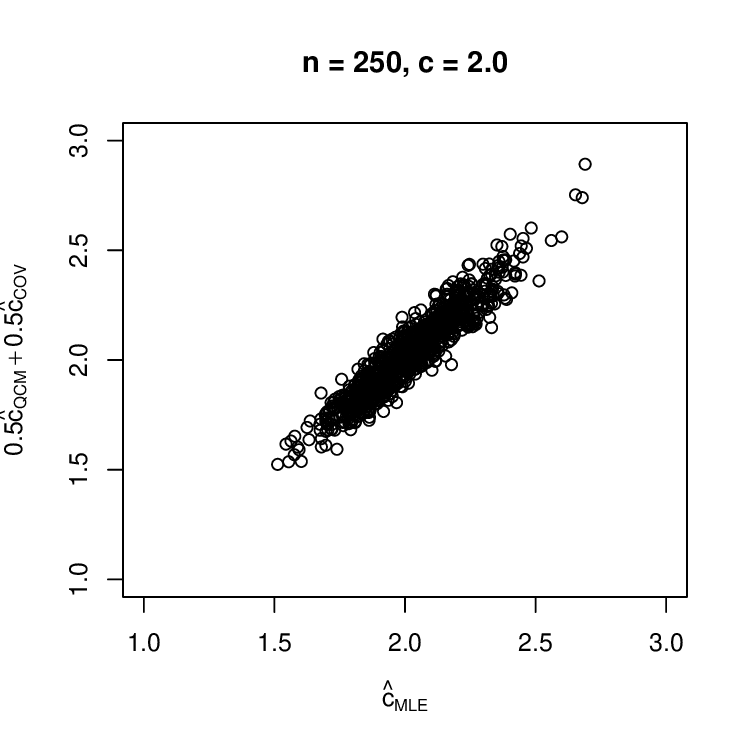}
\end{center}
\caption{Scatterplots of the estimates based on $\hat{c}_{\textrm{QCM}}(0.02,0.48)$, $\hat{c}_{\textrm{MLE}}$, $\hat{c}_{\textrm{COV}}$ and $\tfrac{1}{2}(\hat c_{\textrm{COV}}+\hat c_{\textrm{QCM}}(0.02,0.48))$. The results are based on 1000 strong Monte Carlo samples with $n=250$ and $c=2$. These plots might be interpreted as a graphical representation of the correlation between estimators. While the estimators $\hat{c}_{\textrm{COV}}$ and $\hat{c}_{\textrm{QCM}}$ are rather uncorrelated, it might explain why their linear combination outperforms the results obtained by each estimator separately. What is more, the estimators $\hat{c}_{\textrm{COM}}$ and $\hat{c}_{\textrm{MLE}}$ have a high correlation which might affect the usage of the ratio of those estimators as a test statistic. }
\label{fig:Ccorr}
\end{figure}

\begin{table}[ht]
\centering
\begin{tabular}{|r|rrrr|}
  \hline
  Covariances & $\hat{c}_{QCM}$ & $\hat{c}_{QCV}$ & $\hat{c}_{MLE}$ & $\hat{c}_{COV}$  \\ 
  \hline
$\hat{c}_{QCM}$ & 0.05 & 0.05 & 0.03 & 0.02 \\ 
  $\hat{c}_{QCV}$ & 0.05 & 0.09 & 0.03 & 0.02 \\ 
  $\hat{c}_{MLE}$ & 0.03 & 0.03 & 0.03 & 0.03 \\ 
  $\hat{c}_{COV}$ & 0.02 & 0.02 & 0.03 & 0.06 \\ 
   \hline
\end{tabular}
\begin{tabular}{|r|rrrr|}
  \hline
 Correlations & $\hat{c}_{QCM}$ & $\hat{c}_{QCV}$ & $\hat{c}_{MLE}$ & $\hat{c}_{COV}$  \\ 
  \hline
$\hat{c}_{QCM}$ & 1.00 & 0.85 & 0.83 & 0.40\\ 
  $\hat{c}_{QCV}$ & 0.85 & 1.00 & 0.60 & 0.23\\ 
  $\hat{c}_{MLE}$ & 0.83  & 0.60 & 1.00 & 0.76 \\ 
  $\hat{c}_{COV}$ & 0.40 & 0.23 & 0.76  &1.00 \\ 
   \hline
\end{tabular}

\caption{The covariance and the correlation matrices between four $\hat{c}$ statistic based on strong Monte Carlo sample of size 10\,000, for $c=2$ and $n=250$. As we can see $\hat{c}_{COV}$ and $\hat{c}_{QCM}$ are relatively low correlated with high $\hat{c}_{MLE}$ correlation which might suggest why their combination outperform both of the estimators separately.}\label{tab:correlations}
\end{table}

As before, we get that $T_n$ is asymptotically normal under the null hypothesis.

\begin{proposition}\label{pr:Tn.stat}
Let $X_i$, $i = 1,2, \ldots$, be i.i.d. random variables following $Lv(c)$ with some $c>0$. Also, let $T_n$ be given by~\eqref{eq:Tn}. Then, we get $T_n\overset{d}{\to} \mathcal{N}(0, \tau_T^2)$, as $n\to\infty$, where $\tau_T > 0$ is a constant independent of $c$.
\end{proposition}
Proposition~\ref{pr:Tn.stat} is a direct implication of Proposition~\ref{pr:appendix}, which proof is deferred to the Appendix.  Proposition~\ref{pr:Tn.stat} is illustrated in Figure~\ref{fig:Ccorr2}.

 \begin{figure}[htp!]
\begin{center}
\includegraphics[width=0.26\textwidth]{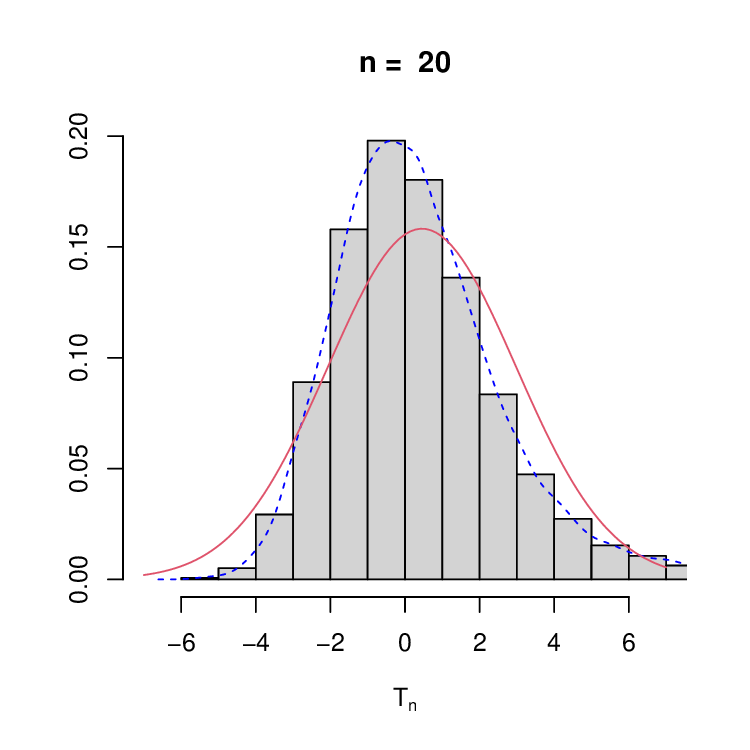}
\includegraphics[width=0.26\textwidth]{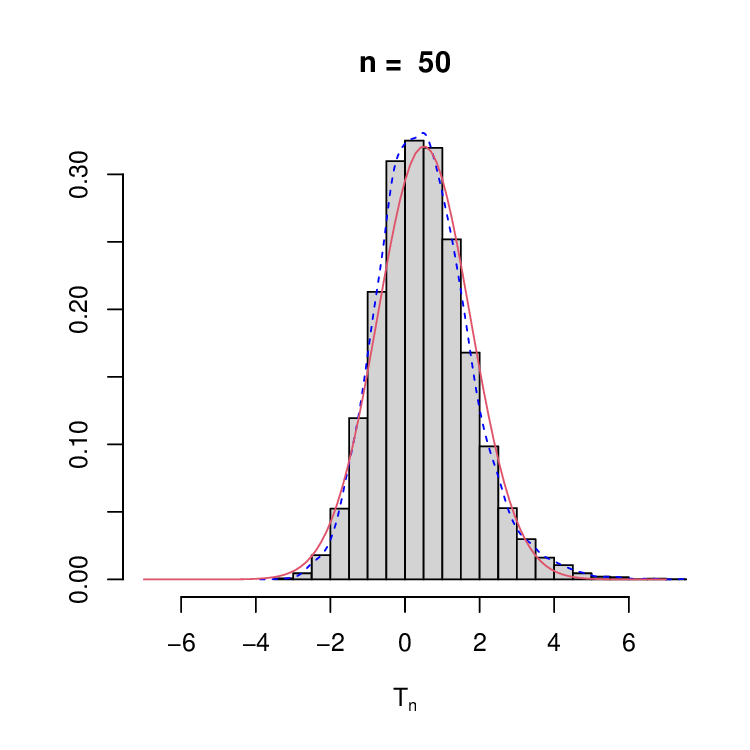}
\includegraphics[width=0.26\textwidth]{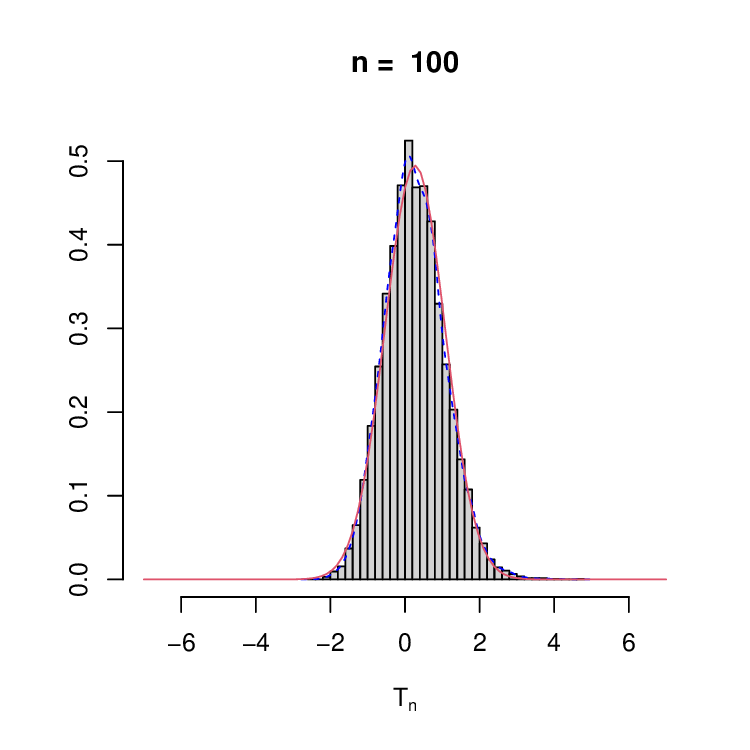}
\includegraphics[width=0.26\textwidth]{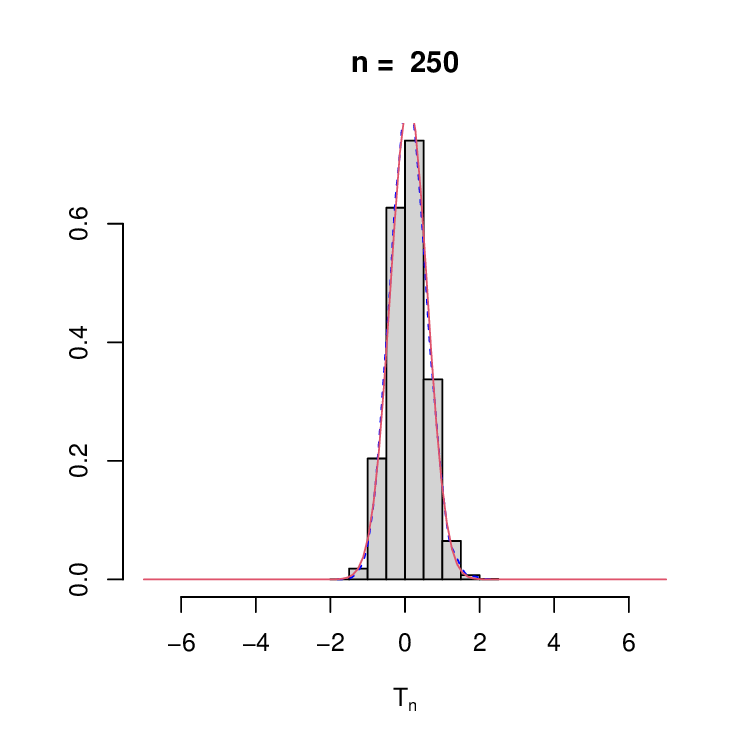}
\includegraphics[width=0.26\textwidth]{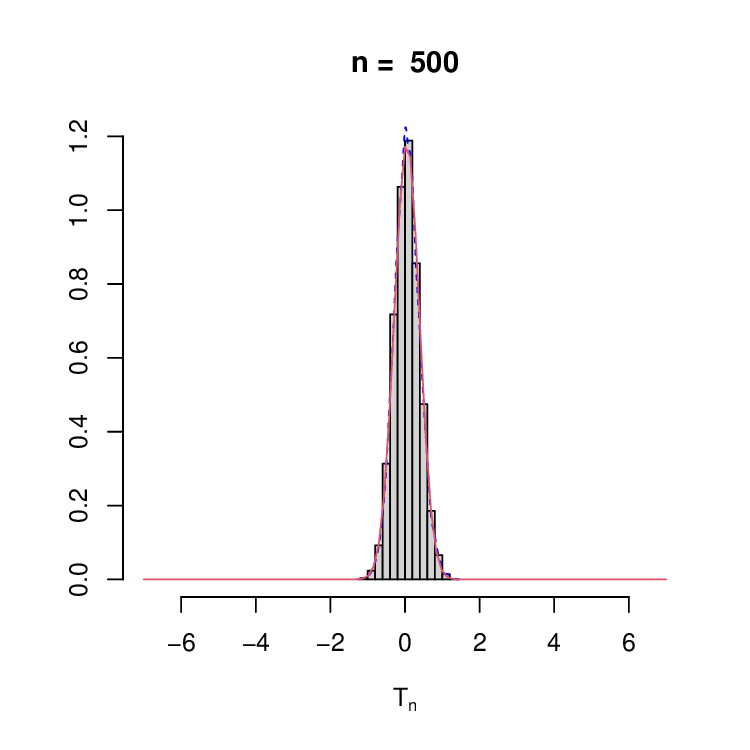}
\includegraphics[width=0.26\textwidth]{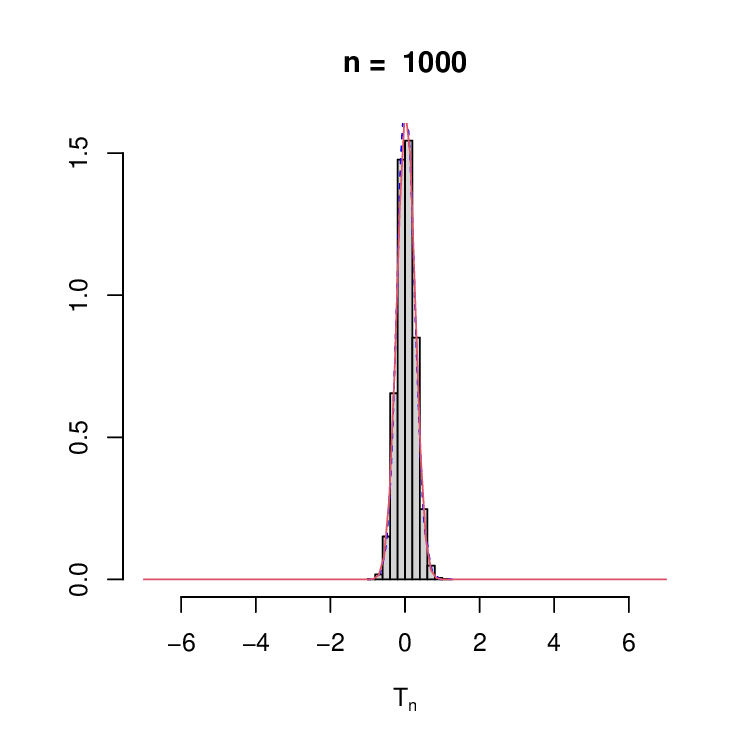}
\end{center}
\caption{Histogram of $T_n$ distribution under the null hypothesis with $n \in \{20,50,100,250,500,1000\}$. The blue dashed line represents the kernel density while the solid red line shows the fitted normal density. Each histogram is based on 10 000 strong Monte Carlo simulations.}
\label{fig:Ccorr2}
\end{figure}

\FloatBarrier
 
\newpage
\section{Test power}\label{sec4}


In this section we present the empirical power analysis for test statistics $V_n$, $O_n$, and $T_n$. Namely, we compare the power of test statistic $V_n$ introduced in \cite{KumBha2022}, with conditional-moment based test statistics $O_n$ and $T_n$, introduced in this paper. For completeness, we also consider two other tests statistics recently introduced in the literature. More specifically, following \cite{Lukic2023} and \cite{kumari2023}, respectively, we set

\begin{align}
    R_{a,n} & := \frac{3\sqrt{\pi}}{4n^2} \sum_{i,j=1}^{n} \left( \frac{1}{ \left(a + \frac{Y_i + Y_j}{4} \right)^{\frac{5}{2}}} - \frac{1}{2(a+Y_i)^{\frac{5}{2}}} - \frac{1}{2(a+Y_j)^{\frac{5}{2}}}\right),\\
    \hat{\Delta}_n &:= \frac{3}{2} \hat{U}_1 - \frac{\hat{c}_{MLE}}{2} \hat{U}_2 - \frac{1}{2},
\end{align}
where $Y_k := \frac{X_k}{\hat{c}_{MLE}}$ for $k\in\{1,\ldots,n\}$, while $\hat{U}_1$ and $\hat{U}_2$ are unbiased estimators of $\mathbb{E}\left( \frac{\min(X_1,X_2)}{X_1} \right)$ and $\mathbb{E}\left( \frac{\min(X_1,X_2)}{X^2_1} \right)$. The parameter $a>0$ in $R_{a,n}$ needs to be adjusted to the specific testing framework as it significantly affects the statistic performance. In the following, for simplicity, we set $a = 0.2$, which results in reasonable good performance for a relatively broad class of alternatives; see Table 3 in~\cite{Lukic2023} for details.

For transparency, most of the alternative distributions were taken from~\cite {KumBha2022}. This choice of alternatives facilitates the power analysis for various aspects of the underlying distribution, including heavy tails and undefined moments. Specific parameter choices and PDFs for alternative distributions are summarised in Table~\ref{tab:ProbDist}, see \cite {KumBha2022} for more details.
\begin{table}[ht!]
\begin{center}
\begin{tabular}{|c| c| c |} 
 \hline
 Distribution & Considered parameters & PDF \\ 
 \hline
 $\textrm{Gamma}(k,\theta)$ & $k = 2$, $\theta = 3 $  &$
 f(x) = \frac{1}{\Gamma(k)\theta^k}x^{k-1}e^{-\frac{x}{\theta}}$  
 \\ 
 \hline
 $\textrm{Chi-Squared}(k)$ & $k = 4$ & $ f(x) = \frac{1}{2^{k/2}\Gamma(k/2)}x^{k/2 -1}e^{\frac{-x}{2}}$  \\
 \hline
 $\textrm{Weibull}(\lambda,k)$ & $\lambda = 1.75$, $k = 1$ & $f(x) = \frac{k}{\lambda}\left(\frac{x}{\lambda}\right)^{k-1}e^{-(x/\lambda)^k}$  \\
 \hline
 $\textrm{Lognormal}(\mu,\sigma^2)$ & $\mu  = 0$, $\sigma^2 = 1$ & $f(x) = \frac{1}{x\sigma\sqrt{2\pi} }\exp\left( -\frac{(\ln(x)-\mu)^2}{2\sigma^2} \right)$  \\
 \hline
  $\textrm{Pareto}(\sigma,\alpha)$ & $\sigma = 0.75$, $\alpha = 1.0$ and $\sigma = 1.5, \alpha = 0.5$  & $f(x) = \frac{\alpha \sigma^\alpha}{x^{\alpha +1}}$  \\
 \hline
  $\textrm{Rayleigh}(\sigma)$ & $\sigma = 1.0$ & $f(x) = \frac{x}{\sigma^2}e^{-x^2/(2\sigma^2)}$  \\
 \hline
  $\textrm{Half-normal}(\sigma)$ & $\sigma =1.0$ & $f(x) = \frac{\sqrt{2}}{\sigma\sqrt{\pi} }\exp\left( -\frac{x^2}{2\sigma^2} \right)$ \\
 \hline
 $\textrm{Fr\'echet}(\alpha, s,m)$ & $\alpha = 0.0$, $s = 0.5$, $m = 1.0$ & $f(x) = \frac{\alpha}{s}\left(\frac{x-m}{x}\right)^{-1-\alpha}\exp\left(-\left(\frac{x-m}{s}\right)^{-\alpha}\right)$ \\
 \hline
  $|\textrm{Log-Gamma}(k,\theta)|$ &  $ k = 3$, $\theta = 2$ & $f(x) = \frac{1}{\Gamma(k)}\big(e^{-\frac{e^x}{\theta}}\left( \frac{e^x}{\theta} \right)^k+e^{-\frac{e^{-x}}{\theta}}\left( \frac{e^{-x}}{\theta} \right)^k\big)$ \\
 \hline
 $\textrm{Inv-Gaussian}(\mu,\lambda)$ & $\mu = 1.0 $, $\lambda = 1.5$ & $f(x) = \sqrt{\frac{\lambda}{2\pi x^3}}\exp \left(-\frac{\lambda(x-\mu)^2}{2\mu^2x} \right)$ \\
 \hline
 $\textrm{Burr}(\mu,\eta,\sigma)$ & $\mu = 1.5$, $\eta = 0.5$, $\sigma = 0.5$  & $f(x) = \frac{\eta\sigma(x/\mu)^{\sigma-1}}{\mu(1+(x/\mu)^{\sigma})^{\eta + 1}}$ \\
 \hline
\end{tabular}
\caption{Probability distribution used in the alternative hypothesis for the power simulation. Note that all probability density functions (PDFs) have $[0,\infty)$ support.}
\label{tab:ProbDist}
\end{center}
\end{table}
\FloatBarrier


For each alternative, the power of the tests was calculated empirically using a standard strong Monte Carlo framework. For any test, we considered two-sided rejection regions and used simulated rejection thresholds obtained using a strong 10 000 samples from the null hypothesis with the underlying sample size and the test size. In all cases, the test power is calculated based on 10 000 strong Monte Carlo sample from the alternative hypothesis. 

The results for the significance level $5\%$ and $1\%$ are presented in Table~\ref{tab:power_0.05} and Table~\ref{tab:power_0.01}, respectively. 
\begin{table}[]
\scalebox{1.0}{
\begin{tabular}{|r|cc|ccc|cc|ccc|cc|ccc|}
\hline
    & \multicolumn{5}{c|}{Gamma(2,3)}                                               & \multicolumn{5}{c|}{Chi-Squared(4)}                                           & \multicolumn{5}{c|}{Weibull(1.75,1)}                                           \\\hline
n   & $O_n$         & $T_n$         & $V_n$         & $R_{0.2,n}$   & $\Delta_n$    & $O_n$         & $T_n$         & $V_n$         & $R_{0.2,n}$   & $\Delta_n$    & $O_n$         & $T_n$         & $V_n$         & $R_{0.2,n}$   & $\Delta_n$     \\ \hline
20  & \textbf{1.00} & 0.92          & 0.82          & 0.75          & 0.99          & \textbf{1.00} & 0.91          & 0.82          & 0.74          & 0.99          & \textbf{1.00} & 0.95          & 0.84          & 0.79          & \textbf{ 1.00} \\
30  & \textbf{1.00} & 0.99          & 0.87          & 0.84          & \textbf{1.00} & \textbf{1.00} & 0.99          & 0.87          & 0.83          & \textbf{1.00} & \textbf{1.00} & \textbf{1.00} & 0.88          & 0.87          & \textbf{1.00}  \\
50  & \textbf{1.00} & \textbf{1.00} & 0.92          & 0.94          & \textbf{1.00} & \textbf{1.00} & \textbf{1.00} & 0.91          & 0.94          & \textbf{1.00} & \textbf{1.00} & \textbf{1.00} & 0.91          & 0.96          & \textbf{1.00}  \\
80  & \textbf{1.00} & \textbf{1.00} & 0.94          & 0.99          & \textbf{1.00} & \textbf{1.00} & \textbf{1.00} & 0.94          & \textbf{1.00} & \textbf{1.00} & \textbf{1.00} & \textbf{1.00} & 0.94          & 0.99          & \textbf{1.00}  \\
100 & \textbf{1.00} & \textbf{1.00} & 0.95          & \textbf{1.00} & \textbf{1.00} & \textbf{1.00} & \textbf{1.00} & 0.95          & \textbf{1.00} & \textbf{1.00} & \textbf{1.00} & \textbf{1.00} & 0.95          & \textbf{1.00} & \textbf{1.00}  \\
120 & \textbf{1.00} & \textbf{1.00} & 0.96          & \textbf{1.00} & \textbf{1.00} & \textbf{1.00} & \textbf{1.00} & 0.96          & \textbf{1.00} & \textbf{1.00} & \textbf{1.00} & \textbf{1.00} & 0.96          & \textbf{1.00} & \textbf{1.00}  \\
250 & \textbf{1.00} & \textbf{1.00} & 0.98          & \textbf{1.00} & \textbf{1.00} & \textbf{1.00} & \textbf{1.00} & 0.98          & \textbf{1.00} & \textbf{1.00} & \textbf{1.00} & \textbf{1.00} & 0.97          & \textbf{1.00} & \textbf{1.00}  \\ \hline
    & \multicolumn{5}{c|}{Pareto(0.75,1.0)}                                         & \multicolumn{5}{c|}{Rayleigh(1.0)}                                            & \multicolumn{5}{c|}{Half-normal(1.0)}                                          \\\hline
n   & $O_n$         & $T_n$         & $V_n$         & $R_{0.2,n}$   & $\Delta_n$    & $O_n$         & $T_n$         & $V_n$         & $R_{0.2,n}$   & $\Delta_n$    & $O_n$         & $T_n$         & $V_n$         & $R_{0.2,n}$   & $\Delta_n$     \\ \hline
20  & 0.55          & 0.80          & 0.85 & \textbf{0.99}          & 0.96          & \textbf{1.00} & 0.99          & 0.92          & 0.90          & 1.00          & \textbf{0.92} & 0.63          & 0.46          & 0.31          & 0.77           \\
30  & 0.72          & 0.92          & \textbf{0.96} & \textbf{1.00} & \textbf{1.00} & \textbf{1.00} & 0.99          & 0.95          & 0.90          & \textbf{1.00} & \textbf{0.99} & 0.81          & 0.50          & 0.31          & 0.89           \\
50  & 0.89          & 0.99          & \textbf{1.00} & \textbf{1.00} & \textbf{1.00} & \textbf{1.00} & \textbf{1.00} & 0.97          & 0.99          & \textbf{1.00} & \textbf{1.00} & 0.97          & 0.55          & 0.35          & 0.98           \\
80  & 0.98          & \textbf{1.00} & \textbf{1.00} & \textbf{1.00} & \textbf{1.00} & \textbf{1.00} & \textbf{1.00} & 0.98          & \textbf{1.00} & \textbf{1.00} & \textbf{1.00} & \textbf{1.00} & 0.60          & 0.38          & \textbf{1.00}  \\
100 & 0.99          & \textbf{1.00} & \textbf{1.00} & \textbf{1.00} & \textbf{1.00} & \textbf{1.00} & \textbf{1.00} & 0.98          & \textbf{1.00} & \textbf{1.00} & \textbf{1.00} & \textbf{1.00} & 0.62          & 0.39          & \textbf{1.00}  \\
120 & \textbf{1.00} & \textbf{1.00} & \textbf{1.00} & \textbf{1.00} & \textbf{1.00} & \textbf{1.00} & \textbf{1.00} & 0.99          & \textbf{1.00} & \textbf{1.00} & \textbf{1.00} & \textbf{1.00} & 0.64          & 0.40          & \textbf{1.00}  \\
250 & \textbf{1.00} & \textbf{1.00} & \textbf{1.00} & \textbf{1.00} & \textbf{1.00} & \textbf{1.00} & \textbf{1.00} & 0.99          & \textbf{1.00} & \textbf{1.00} & \textbf{1.00} & \textbf{1.00} & 0.75          & 0.61          & \textbf{1.00}  \\ \hline
    & \multicolumn{5}{c|}{$|Log-Gamma(3, 2)|$}                                      & \multicolumn{5}{c|}{Pareto(1.5,0.5)}                                          & \multicolumn{5}{c|}{Inv-Gaussian(1.0,1.5)}                                     \\\hline
n   & $O_n$         & $T_n$         & $V_n$         & $R_{0.2,n}$   & $\Delta_n$    & $O_n$         & $T_n$         & $V_n$         & $R_{0.2,n}$   & $\Delta_n$    & $O_n$         & $T_n$         & $V_n$         & $R_{0.2,n}$   & $\Delta_n$     \\ \hline
20  & \textbf{0.98} & 0.78          & 0.55          & 0.44          & 0.93          & \textbf{1.00} & \textbf{1.00} & \textbf{1.00} & \textbf{1.00} & \textbf{1.00} & \textbf{0.91} & 0.90          & 0.87          & 0.99          & \textbf{1.00}  \\
30  & \textbf{1.00} & 0.93          & 0.58          & 0.46          & 0.99          & \textbf{1.00} & \textbf{1.00} & \textbf{1.00} & \textbf{1.00} & \textbf{1.00} & 0.99 & 0.99 & 0.98          & \textbf{1.00} & \textbf{1.00}  \\
50  & \textbf{1.00} & 0.99          & 0.62          & 0.48          & \textbf{1.00} & \textbf{1.00} & \textbf{1.00} & \textbf{1.00} & \textbf{1.00} & \textbf{1.00} & \textbf{1.00} & \textbf{1.00} & \textbf{1.00} & \textbf{1.00} & \textbf{1.00}  \\
80  & \textbf{1.00} & \textbf{1.00} & 0.65          & 0.49          & \textbf{1.00} & \textbf{1.00} & \textbf{1.00} & \textbf{1.00} & \textbf{1.00} & \textbf{1.00} & \textbf{1.00} & \textbf{1.00} & \textbf{1.00} & \textbf{1.00} & \textbf{1.00}  \\
100 & \textbf{1.00} & \textbf{1.00} & 0.66          & 0.57          & \textbf{1.00} & \textbf{1.00} & \textbf{1.00} & \textbf{1.00} & \textbf{1.00} & \textbf{1.00} & \textbf{1.00} & \textbf{1.00} & \textbf{1.00} & \textbf{1.00} & \textbf{1.00}  \\
120 & \textbf{1.00} & \textbf{1.00} & 0.68          & 0.61          & \textbf{1.00} & \textbf{1.00} & \textbf{1.00} & \textbf{1.00} & \textbf{1.00} & \textbf{1.00} & \textbf{1.00} & \textbf{1.00} & \textbf{1.00} & \textbf{1.00} & \textbf{1.00}  \\
250 & \textbf{1.00} & \textbf{1.00} & 0.75          & 0.84          & \textbf{1.00} & \textbf{1.00} & \textbf{1.00} & \textbf{1.00} & \textbf{1.00} & \textbf{1.00} & \textbf{1.00} & \textbf{1.00} & \textbf{1.00} & \textbf{1.00} & \textbf{1.00}  \\ \hline
    & \multicolumn{5}{c|}{Lognormal(0,1)}                                           & \multicolumn{5}{c|}{Fréchet(0.0,0.5,1.0)}                                     & \multicolumn{5}{c|}{Burr(1.5,0.5,0.5)}                                         \\\hline
n   & $O_n$         & $T_n$         & $V_n$         & $R_{0.2,n}$   & $\Delta_n$    & $O_n$         & $T_n$         & $V_n$         & $R_{0.2,n}$   & $\Delta_n$    & $O_n$         & $T_n$         & $V_n$         & $R_{0.2,n}$   & $\Delta_n$     \\ \hline
20  & \textbf{0.89} & 0.75          & 0.70          & 0.55          & 0.85          & 0.51          & 0.60          & \textbf{0.62} & 0.53          & 0.57          & 0.84          & 0.80          & \textbf{1.00} & 0.49          & 0.99           \\
30  & \textbf{0.99} & 0.92          & 0.82          & 0.69          & 0.98          & 0.72          & \textbf{0.82} & \textbf{0.82} & 0.71          & 0.78          & 0.96          & 0.90          & \textbf{1.00} & 0.58          & \textbf{1.00}  \\
50  & \textbf{1.00} & 0.99          & 0.93          & 0.88          & \textbf{1.00} & 0.93          & \textbf{0.97} & \textbf{0.97} & 0.91          & 0.96          & \textbf{1.00} & 0.97          & \textbf{1.00} & 0.61          & \textbf{1.00}  \\
80  & \textbf{1.00} & \textbf{1.00} & 0.98          & 0.97          & \textbf{1.00} & 0.99          & \textbf{1.00} & \textbf{1.00} & \textbf{1.00} & \textbf{1.00} & \textbf{1.00} & \textbf{1.00} & \textbf{1.00} & 0.67          & \textbf{1.00}  \\
100 & \textbf{1.00} & \textbf{1.00} & 0.99          & \textbf{1.00} & \textbf{1.00} & \textbf{1.00} & \textbf{1.00} & \textbf{1.00} & \textbf{1.00} & \textbf{1.00} & \textbf{1.00} & \textbf{1.00} & \textbf{1.00} & 0.68          & \textbf{1.00}  \\
120 & \textbf{1.00} & \textbf{1.00} & 0.99          & \textbf{1.00} & \textbf{1.00} & \textbf{1.00} & \textbf{1.00} & \textbf{1.00} & \textbf{1.00} & \textbf{1.00} & \textbf{1.00} & \textbf{1.00} & \textbf{1.00} & 0.99          & \textbf{1.00}  \\
250 & \textbf{1.00} & \textbf{1.00} & \textbf{1.00} & \textbf{1.00} & \textbf{1.00} & \textbf{1.00} & \textbf{1.00} & \textbf{1.00} & \textbf{1.00} & \textbf{1.00} & \textbf{1.00} & \textbf{1.00} & \textbf{1.00} & \textbf{1.00}  & \textbf{1.00}  \\ \hline
\end{tabular}
\caption{Power of the tests, obtained by empirical distribution given by 10\,000 strong Monte Carlo simulations for each alternative hypothesis at significance level $5\%$.}\label{tab:power_0.05}}
\end{table}
\begin{table}[]
\scalebox{1.0}{
\begin{tabular}{|r|cc|ccc|cc|ccc|cc|ccc|}
\hline
    & \multicolumn{5}{c|}{Gamma(2,3)}                                               & \multicolumn{5}{c|}{Chi-Squared(4)}                                           & \multicolumn{5}{c|}{Weibull(1.75,1)}                                          \\\hline
n   & $O_n$         & $T_n$         & $V_n$         & $R_{0.2,n}$   & $\Delta_n$    & $O_n$         & $T_n$         & $V_n$         & $R_{0.2,n}$   & $\Delta_n$    & $O_n$         & $T_n$         & $V_n$         & $R_{0.2,n}$   & $\Delta_n$    \\ \hline
20  & \textbf{0.96} & 0.73          & 0.67          & 0.6           & 0.94          & \textbf{0.96} & 0.72          & 0.67          & 0.58          & 0.94          & \textbf{0.99} & 0.78          & 0.73          & 0.68          & 0.98          \\
30  & \textbf{1.00} & 0.88          & 0.77          & 0.72          & \textbf{1.00} & \textbf{1.00} & 0.88          & 0.77          & 0.71          & \textbf{1.00} & \textbf{1.00} & 0.91          & 0.81          & 0.79          & \textbf{1.00} \\
50  & \textbf{1.00} & 0.99          & 0.87          & 0.88          & \textbf{1.00} & \textbf{1.00} & 0.99          & 0.87          & 0.71          & \textbf{1.00} & \textbf{1.00} & \textbf{1.00} & 0.87          & 0.91          & \textbf{1.00} \\
80  & \textbf{1.00} & \textbf{1.00} & 0.91          & 0.97          & \textbf{1.00} & \textbf{1.00} & \textbf{1.00} & 0.91          & 0.97          & \textbf{1.00} & \textbf{1.00} & \textbf{1.00} & 0.91          & 0.98          & \textbf{1.00} \\
100 & \textbf{1.00} & \textbf{1.00} & 0.93          & 0.99          & \textbf{1.00} & \textbf{1.00} & \textbf{1.00} & 0.93          & 0.97          & \textbf{1.00} & \textbf{1.00} & \textbf{1.00} & 0.93          & 0.99          & \textbf{1.00} \\
120 & \textbf{1.00} & \textbf{1.00} & 0.94          & \textbf{1.00} & \textbf{1.00} & \textbf{1.00} & \textbf{1.00} & 0.95          & \textbf{1.00} & \textbf{1.00} & \textbf{1.00} & \textbf{1.00} & 0.94          & \textbf{1.00} & \textbf{1.00} \\
250 & \textbf{1.00} & \textbf{1.00} & 0.97          & \textbf{1.00} & \textbf{1.00} & \textbf{1.00} & \textbf{1.00} & 0.97          & \textbf{1.00} & \textbf{1.00} & \textbf{1.00} & \textbf{1.00} & 0.96          & \textbf{1.00} & \textbf{1.00} \\ \hline
    & \multicolumn{5}{c|}{Pareto(0.75,1.0)}                                         & \multicolumn{5}{c|}{Rayleigh(1.0)}                                            & \multicolumn{5}{c|}{Half-normal(1.0)}                                         \\\hline
n   & $O_n$         & $T_n$         & $V_n$         & $R_{0.2,n}$   & $\Delta_n$    & $O_n$         & $T_n$         & $V_n$         & $R_{0.2,n}$   & $\Delta_n$    & $O_n$         & $T_n$         & $V_n$         & $R_{0.2,n}$   & $\Delta_n$    \\ \hline
20  & 0.31          & 0.58          & 0.61          & \textbf{0.94}          & 0.85          & \textbf{1.00} & 0.90          & 0.86          & 0.83          & \textbf{1.00} & \textbf{0.68} & 0.39          & 0.30          & 0.18          & 0.53          \\
30  & 0.53          & 0.82          & 0.85 & \textbf{1.00} & 0.98          & \textbf{1.00} & 0.98          & 0.91          & 0.91          & \textbf{1.00} & \textbf{0.92} & 0.61          & 0.35          & 0.19          & 0.75          \\
50  & 0.75          & 0.97          & \textbf{0.99} & \textbf{1.00} & \textbf{1.00} & \textbf{1.00} & \textbf{1.00} & 0.95          & 0.98          & \textbf{1.00} & \textbf{1.00} & 0.97          & 0.42          & 0.20          & 0.93          \\
80  & 0.93          & \textbf{1.00} & \textbf{1.00} & \textbf{1.00} & \textbf{1.00} & \textbf{1.00} & \textbf{1.00} & 0.97          & \textbf{1.00} & \textbf{1.00} & \textbf{1.00} & \textbf{1.00} & 0.49          & 0.24          & 0.99          \\
100 & 0.96          & \textbf{1.00} & \textbf{1.00} & \textbf{1.00} & \textbf{1.00} & \textbf{1.00} & \textbf{1.00} & 0.98          & \textbf{1.00} & \textbf{1.00} & \textbf{1.00} & \textbf{1.00} & 0.52          & 0.27          & \textbf{1.00} \\
120 & 0.99          & \textbf{1.00} & \textbf{1.00} & \textbf{1.00} & \textbf{1.00} & \textbf{1.00} & \textbf{1.00} & 0.98          & \textbf{1.00} & \textbf{1.00} & \textbf{1.00} & \textbf{1.00} & 0.54          & 0.31          & \textbf{1.00} \\
250 & \textbf{1.00} & \textbf{1.00} & \textbf{1.00} & \textbf{1.00} & \textbf{1.00} & \textbf{1.00} & \textbf{1.00} & 0.99          & \textbf{1.00} & \textbf{1.00} & \textbf{1.00} & \textbf{1.00} & 0.68          & 0.47          & \textbf{1.00} \\ \hline
    & \multicolumn{5}{c|}{$|Log-Gamma(3, 2)|$}                                      & \multicolumn{5}{c|}{Pareto(1.5,0.5)}                                          & \multicolumn{5}{c|}{Inv-Gaussian(1.0,1.5)}                                    \\\hline
n   & $O_n$         & $T_n$         & $V_n$         & $R_{0.2,n}$   & $\Delta_n$    & $O_n$         & $T_n$         & $V_n$         & $R_{0.2,n}$   & $\Delta_n$    & $O_n$         & $T_n$         & $V_n$         & $R_{0.2,n}$   & $\Delta_n$    \\ \hline
20  & \textbf{0.89} & 0.51          & 0.40          & 0.32          & 0.81          & 0.99          & \textbf{1.00} & \textbf{1.00} & \textbf{1.00} & \textbf{1.00} & 0.56          & 0.61 & 0.59          & 0.95          & \textbf{0.99 }         \\
30  & \textbf{0.99} & 0.76          & 0.44          & 0.33          & 0.95          & \textbf{1.00} & \textbf{1.00} & \textbf{1.00} & \textbf{1.00} & \textbf{1.00} & 0.90          & 0.92 & 0.86          & \textbf{1.00} & \textbf{1.00} \\
50  & \textbf{1.00} & \textbf{1.00} & 0.50          & 0.34          & \textbf{1.00} & \textbf{1.00} & \textbf{1.00} & \textbf{1.00} & \textbf{1.00} & \textbf{1.00} & \textbf{1.00} & \textbf{1.00} & 0.99          & \textbf{1.00} & \textbf{1.00} \\
80  & \textbf{1.00} & \textbf{1.00} & 0.53          & 0.35          & \textbf{1.00} & \textbf{1.00} & \textbf{1.00} & \textbf{1.00} & \textbf{1.00} & \textbf{1.00} & \textbf{1.00} & \textbf{1.00} & \textbf{1.00} & \textbf{1.00} & \textbf{1.00} \\
100 & \textbf{1.00} & \textbf{1.00} & 0.56          & 0.44          & \textbf{1.00} & \textbf{1.00} & \textbf{1.00} & \textbf{1.00} & \textbf{1.00} & \textbf{1.00} & \textbf{1.00} & \textbf{1.00} & \textbf{1.00} & \textbf{1.00} & \textbf{1.00} \\
120 & \textbf{1.00} & \textbf{1.00} & 0.57          & \textbf{0.48} & \textbf{1.00} & \textbf{1.00} & \textbf{1.00} & \textbf{1.00} & \textbf{1.00} & \textbf{1.00} & \textbf{1.00} & \textbf{1.00} & \textbf{1.00} & \textbf{1.00} & \textbf{1.00} \\
250 & \textbf{1.00} & \textbf{1.00} & 0.67          & 0.74          & \textbf{1.00} & \textbf{1.00} & \textbf{1.00} & \textbf{1.00} & \textbf{1.00} & \textbf{1.00} & \textbf{1.00} & \textbf{1.00} & \textbf{1.00} & \textbf{1.00} & \textbf{1.00} \\ \hline
    & \multicolumn{5}{c|}{Lognormal(0,1)}                                           & \multicolumn{5}{c|}{Fréchet(0.0,0.5,1.0)}                                     & \multicolumn{5}{c|}{Burr(1.5,0.5,0.5)}                                        \\\hline
n   & $O_n$         & $T_n$         & $V_n$         & $R_{0.2,n}$   & $\Delta_n$    & $O_n$         & $T_n$         & $V_n$         & $R_{0.2,n}$   & $\Delta_n$    & $O_n$         & $T_n$         & $V_n$         & $R_{0.2,n}$   & $\Delta_n$    \\ \hline
20  & \textbf{0.57}          & 0.48          & 0.44          & 0.32          & 0.56          & 0.23          & \textbf{0.31} & 0.30          & 0.26          & 0.28          & 0.59          & 0.89          & \textbf{0.99} & 0.40          & 0.98          \\
30  & \textbf{0.90}          & 0.76          & 0.63          & 0.48          & 0.87          & 0.47          & \textbf{0.59} & 0.54          & 0.44          & 0.54          & 0.89          & 0.97          & \textbf{1.00} & 0.43          & \textbf{1.00} \\
50  & \textbf{1.00} & 0.97          & 0.85          & 0.73          & \textbf{1.00} & 0.78          & \textbf{0.88} & 0.87          & 0.75          & 0.83          & 0.99          & \textbf{1.00} & \textbf{1.00} & 0.50          & \textbf{1.00} \\
80  & \textbf{1.00} & \textbf{1.00} & 0.95          & 0.89          & \textbf{1.00} & 0.96          & \textbf{0.99} & \textbf{0.99}          & 0.93          & 0.99          & \textbf{1.00} & \textbf{1.00} & \textbf{1.00} & 0.58          & \textbf{1.00} \\
100 & \textbf{1.00} & \textbf{1.00} & 0.97          & 0.97          & \textbf{1.00} & 0.98          & \textbf{1.00} & \textbf{1.00} & 0.99          & \textbf{1.00} & \textbf{1.00} & \textbf{1.00} & \textbf{1.00} & 0.61          & \textbf{1.00} \\
120 & \textbf{1.00} & \textbf{1.00} & 0.99          & 0.99          & \textbf{1.00} & \textbf{1.00} & \textbf{1.00} & \textbf{1.00} & \textbf{1.00} & \textbf{1.00} & \textbf{1.00} & \textbf{1.00} & \textbf{1.00} & 0.72          & \textbf{1.00} \\
250 & \textbf{1.00} & \textbf{1.00} & \textbf{1.00} & \textbf{1.00} & \textbf{1.00} & \textbf{1.00} & \textbf{1.00} & \textbf{1.00} & \textbf{1.00} & \textbf{1.00} & \textbf{1.00} & \textbf{1.00} & \textbf{1.00} & 0.96          & \textbf{1.00} \\ \hline
\end{tabular}
\caption{Power of the tests, obtained by empirical distribution given by 10\,000 strong Monte Carlo simulations for each alternative hypothesis at significance level $1\%$.}\label{tab:power_0.01}}
\end{table} 
From the tables we can see that, despite its simplicity, $O_n$ test provides a satisfactory test power in most of the cases. In particular, $O_n$ outperforms $V_n$ test for almost all of the alternative hypotheses. This is particularly visible for the Half-normal(1.0) distribution and the significance level $0.01$, where the power of the $O_n$ test is more than twice as large as the power of the $V_n$ test, and $O_n$ outperforms all benchmark tests. To better explain the differences in the performance of the statistics, in Figure~\ref{fig:PowerHist} we present the empirical distributions of $O_n$, $T_n$, and $V_n$ test statistics for the null and selected alternative hypothesis. Note that the better the test statistic is, the more effectively it distinguishes distributions under different hypotheses. 

 \begin{figure}[htp!]
\begin{center}
\includegraphics[width=0.32\textwidth]{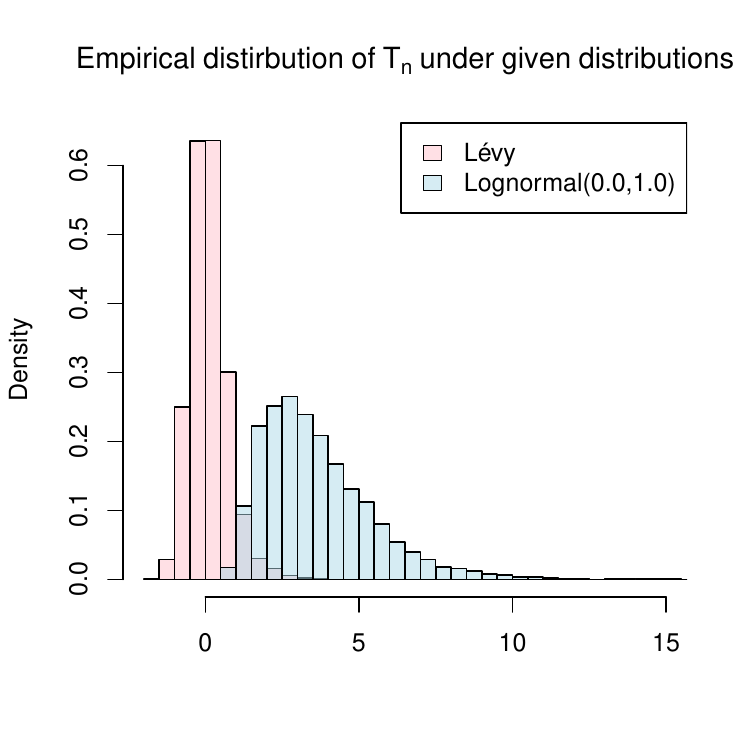}
\includegraphics[width=0.32\textwidth]{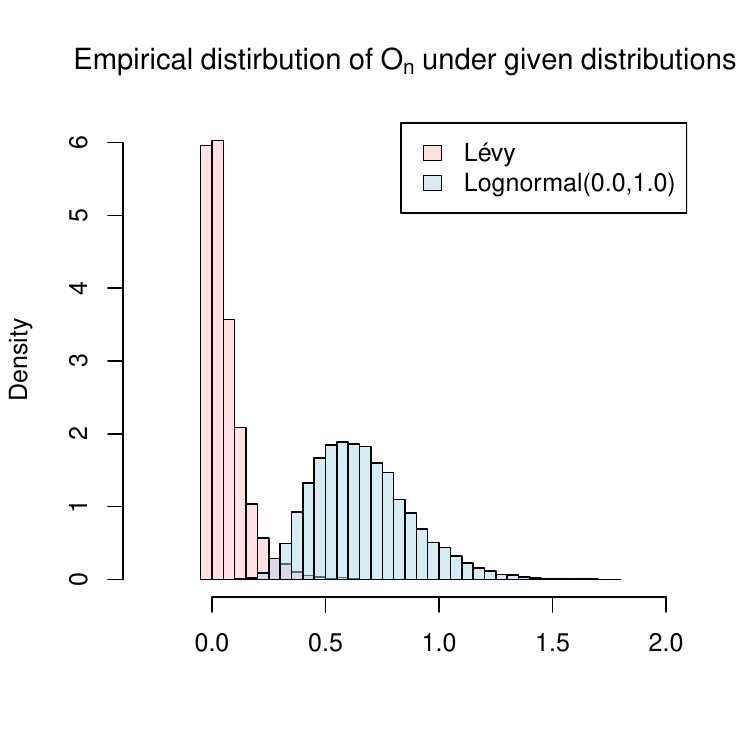}
\includegraphics[width=0.32\textwidth]{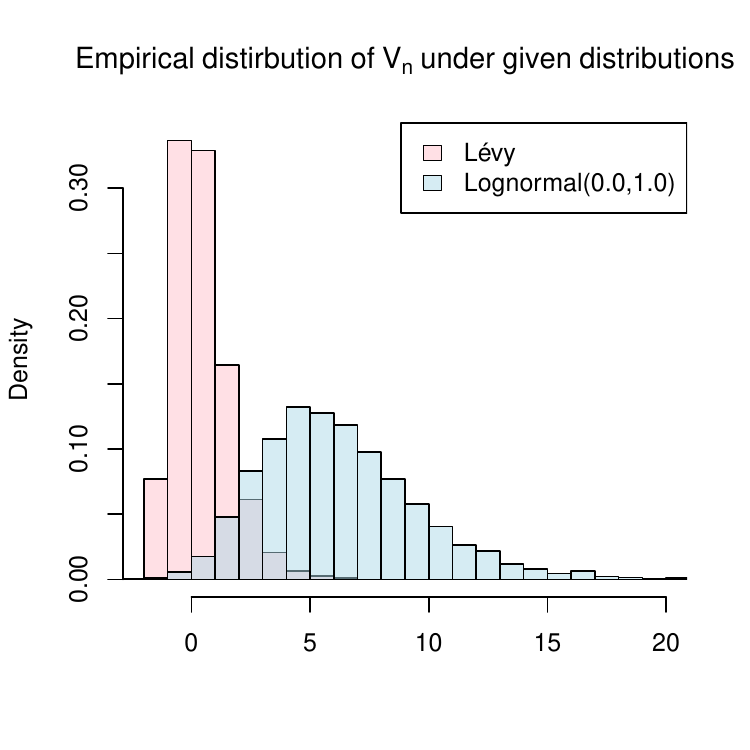}
\end{center}
\caption{Empirical distributions of the tests $V_n$, $T_n$, and $O_n$ under the null (one-sided L\'{e}vy distribution) hypothesis and the exemplary alternative hypothesis ($\textrm{Lognormal}(0,1)$ distribution). The results are based on 10 000 strong Monte Carlo samples with $n=30$. Note that the better test statistics more effectively separate distributions under the different hypotheses.}
\label{fig:PowerHist}
\end{figure}

To sum up, we see that test statistics $O_n$ and $T_n$ frequently outperforms test statistics $V_n$ and $R_{0.2,n}$, and give comparable results with $\Delta_n$, in most of the considered cases. Note that in~\cite{KumBha2022} it has been shown that $V_n$ outperforms the {\it (adjusted) jackknife empirical likelihood} method introduced in \cite{BhaKat2020}, which is the reason we did not include the corresponding test statistics in our study. Consequently, despite their generic formulation based on (sample) conditional moments, test statistics $O_n$ and $T_n$ could be used to complement or refine existing goodness-of-fit testing frameworks.

\section{Location invariant goodness-of-fit test}\label{sec5}

One of the key drawbacks of all considered test statistics, i.e. $O_n$, $T_n$,  $V_n$, $R_{a,n}$, and $\Delta_n$, is the fact that they are not translation invariant and hence tailored only to cover that case when the location is fixed, i.e. the lower support boundary is set as zero. To show how to remediate this, in this section we introduce a location invariant goodness-of-fit test that can be used to check if data follows  a generic one-sided L\'{e}vy distribution in the location-scale parametrisation, i.e. where the null hypothesis is enriched to cover the distributions of the shifted $X+\mu$ random variable, where $\mu\in \mathbb{R}$ and $X\sim Lv(c)$, $c>0$. Please recall that the class of such distributions are generically denoted by $L(\mu,c)$, see Section~\ref{S:preliminaries}.

Again, it should be emphasized that all five test statistics $V_n$, $O_n$, $R_{a,n}$, $\Delta_n$, and $T_n$ are not location invariant, i.e. they cannot be directly applied to a framework where the location parameter $\mu$ is unknown. In particular,  inverse transform and logarithmic transform, which are used in the definition of $V_n$ and $T_n$ statistics are not linear with respect to the location change. In fact, due to the usage of the logarithm transform, these statistics can only be applied to non-negative data.

To overcome the lack of location invariance, one may simply apply some data pre-normalisation transformation, e.g. based on the shift of all observations by their minimum value. However, this greatly complicates the analysis and non-trivially impacts the test performance. Instead, we propose a new test statistic which is given by the ratio of suitable QCVs. More specifically, we consider a new test statistic given by
\begin{align}
C_n:=\sqrt{ n} \left( \frac{\hat{c}_{\textrm{QCV}}(a_1,b_1)}{\hat{c}_{\textrm{QCV}}(a_2,b_2)} - 1\right),
\end{align}
where $0\leq a_1<b_1<1$ and $0\leq a_2<b_2<1$ are some fixed quantile splits.
Note that, for any choice of $a_1,b_1, a_2, b_2$, the statistic $C_n$ is both scale and location invariant. In the following, for illustration purposes, we choose $a_1 = 0.0$, $b_1 = 0.4$, $a_2 = 0.8$, $b_2 = 0.95$.

To check the performance of $C_n$ we use the same alternative distributions which were considered in Section~\ref{sec4}. The results for the significance levels  $5\%$ and $1\%$ are presented in Table~\ref{tab:loc_test_5} and Table~\ref{tab:loc_test_1}, respectively. Note that they are based on 10 000 strong Monte Carlo samples and simulated rejection levels. In all cases, any shift applied to the considered alternatives (or null distribution) will not impact the test power due to location invariance of the test statistic $C_n$.

\begin{table}[ht]
\centering
\begin{tabular}{|l|rrrrrrr|}\hline
Alternative \textbackslash \, n & 20 & 30 & 50 & 80 & 100 & 120 & 250 \\\hline
Gamma(2,3) & 0.95 & 1.00 & 1.00 & 1.00 & 1.00 & 1.00 & 1.00 \\ 
  Chi-Squared(4) & 0.95 & 1.00 & 1.00 & 1.00 & 1.00 & 1.00 & 1.00 \\ 
  Weibull(1.75,1) & 0.99 & 1.00 & 1.00 & 1.00 & 1.00 & 1.00 & 1.00 \\ 
  Lognormal(0,1) & 0.59 & 0.84 & 0.99 & 1.00 & 1.00 & 1.00 & 1.00 \\ 
  Pareto(0.75,1.0) & 0.08 & 0.10 & 0.18 & 0.32 & 0.39 & 0.47 & 0.81 \\ 
  Rayleigh(1.0) & 1.00 & 1.00 & 1.00 & 1.00 & 1.00 & 1.00 & 1.00 \\ 
  Half-normal(1.0) & 0.96 & 1.00 & 1.00 & 1.00 & 1.00 & 1.00 & 1.00 \\ 
  Fr\'{e}chet(0.0,0.5,1.0) & 0.20 & 0.31 & 0.57 & 0.85 & 0.92 & 0.96 & 1.00 \\ 
  $|\textrm{Log-Gamma}(3,2)|$ & 1.00 & 1.00 & 1.00 & 1.00 & 1.00 & 1.00 & 1.00 \\ 
  Pareto(1.5,0.5) & 0.28 & 0.44 & 0.78 & 0.98 & 0.99 & 1.00 & 1.00 \\ 
  Inv-Gaussian(1.0,1.5) & 0.50 & 0.76 & 0.99 & 1.00 & 1.00 & 1.00 & 1.00 \\ 
  Burr(1.5,0.5,0.5) & 0.60 & 0.74 & 0.91 & 0.98 & 0.99 & 1.00 & 1.00 \\\hline
\end{tabular}
\caption{Power of the test $C_n$, obtained by empirical distribution given by 10\,000 strong Monte Carlo simulations at significance level $5\%$. }\label{tab:loc_test_5} 
\end{table}

\begin{table}[ht]
\centering
\begin{tabular}{|l|rrrrrrr|}\hline
Alternative \textbackslash \, n & 20 & 30 & 50 & 80 & 100 & 120 & 250 \\\hline
Gamma(2,3) & 0.61 & 0.93 & 1.00 & 1.00 & 1.00 & 1.00 & 1.00 \\
Chi-Squared(4) & 0.62 & 0.93 & 1.00 & 1.00 & 1.00 & 1.00 & 1.00 \\
Weibull(1.75,1) & 0.80 & 0.99 & 1.00 & 1.00 & 1.00 & 1.00 & 1.00 \\
Lognormal(0,1) & 0.21 & 0.44 & 0.95 & 1.00 & 1.00 & 1.00 & 1.00 \\
Pareto(0.75,1.0) & 0.02 & 0.03 & 0.06 & 0.11 & 0.16 & 0.22 & 0.60 \\
Rayleigh(1.0) & 0.87 & 1.00 & 1.00 & 1.00 & 1.00 & 1.00 & 1.00 \\
Half-normal(1.0) & 0.65 & 0.95 & 1.00 & 1.00 & 1.00 & 1.00 & 1.00 \\
Fr\'{e}chet(0.0,0.5,1.0) & 0.05 & 0.10 & 0.31 & 0.58 & 0.72 & 0.85 & 1.00 \\
$|\textrm{Log-Gamma(3,2)}|$ & 0.93 & 1.00 & 1.00 & 1.00 & 1.00 & 1.00 & 1.00 \\
Pareto(1.5,0.5) & 0.08 & 0.15 & 0.50 & 0.87 & 0.95 & 0.99 & 1.00 \\
Inv-Gaussian(1.0,1.5) & 0.17 & 0.35 & 0.91 & 1.00 & 1.00 & 1.00 & 1.00 \\
Burr(1.5,0.5,0.5) & 0.35 & 0.58 & 0.80 & 0.95 & 0.98 & 0.99 & 1.00 \\\hline
\end{tabular}
\caption{Power of the test $C_n$, obtained by empirical distribution given by 10\,000 strong Monte Carlo simulations at significance level $1\%$. }\label{tab:loc_test_1} 
\end{table}

For a substantial share of alternatives, the power of $C_n$ is comparable to the tests considered in Section~\ref{sec4}. However, for some alternatives, including Pareto and Fr\'{e}chet distributions, the power significantly deteriorates. This can be traced back to the fact that $C_n$ is created to work in the environment with a higher level of  uncertainty, i.e. where both location and scale are unknown. Nevertheless, it should be highlighted again that the tests considered in Section~\ref{sec4} cannot be directly applied in such a framework due to the lack of location invariance and the data non-negativity assumption. To sum up, we believe that the $C_n$ statistic is a promising alternative for generic frameworks but we leave a more systematic study of location invariant tests for future research since the main focus of this paper is set on the development of the framework for the full non-negative support.

\section{Application}\label{sec6}
In this section we show how to apply the proposed tests to real data analysis. We focus on the case study considered  in~\cite{KumBha2022}, where the benchmark test $V_n$ is introduced. The first data set, taken from \cite{Keating1990}, presents the failure time (in hours) of vessels constructed of fiber/epoxy composite materials exposed to high pressure ($n = 20$). According to \cite{Keating1990}, this data set might be modelled by the Gamma distribution with the shape parameter close to 0.5. Therefore, the inverse of these observations should follow the L\'evy distribution. The second data set comes from \url{www.data.gov.in} and describes the average rainfall in January (in mm) for India, for years from 1981 to 2011 ($n = 31$). For completeness, data set 1 and data set 2 are presented in Table~\ref{tab:data1} and~\ref{tab:data2}, respectively.

\begin{table}[ht]
\centering
\begin{tabular}{|rrrrrrrrrr|}
\hline
 274.00 & 1.70 & 871.00 & 1311.00 & 236.00 & 458.00 & 54.90 & 1787.00 & 0.75 & 776.00 \\ \hline
 28.50 & 20.80 & 363.00 & 1661.00 & 828.00 & 290.00 & 175.00 & 970.00 & 1278.00 & 126.00 \\ 
   \hline
\end{tabular}

\caption{Entries for data set 1. The data present failure time (in hours) of vessels constructed of fiber/epoxy composite materials.}\label{tab:data1}
\end{table}

\begin{table}[ht]
\centering
\begin{tabular}{|rr|rr|rr|}
  \hline
 Year & Rainfall& Year & Rainfall& Year & Rainfall\\ 
  \hline
  1981 & 29.30  & 1992 & 16.00 & 2003 & 7.60    \\ 
  1982 & 23.80 & 1993 & 18.20 & 2004 & 25.70   \\ 
  1983 & 18.50  & 1994 & 25.00 & 2005 & 28.10  \\ 
  1984 & 19.00 & 1995& 31.30 &  2006   & 17.70  \\ 
  1985 & 23.20 & 1996& 22.90 & 2007  & 1.70  \\ 
  1986 & 15.50 &  1997& 14.30 & 2008   & 18.40  \\ 
  1987 & 13.20 & 1998 & 16.40 & 2009   & 12.00  \\ 
  1988 & 10.40 & 1999 & 13.70 & 2010 & 7.50 \\ 
  1989 & 15.40& 2000 & 18.40 & 2011   & 6.80 \\ 
  1990 & 16.00&  2001 & 7.30 & &\\ 
  1991 & 14.30 &  2002 & 15.70 & &\\ 
   \hline
\end{tabular}
\caption{Entries for data set 2. The data describe the average rainfall (in mm) in India in each January from 1981 to 2011.}\label{tab:data2}
\end{table}

\begin{figure}[htp!]
\begin{center}
\includegraphics[width=0.33\textwidth]{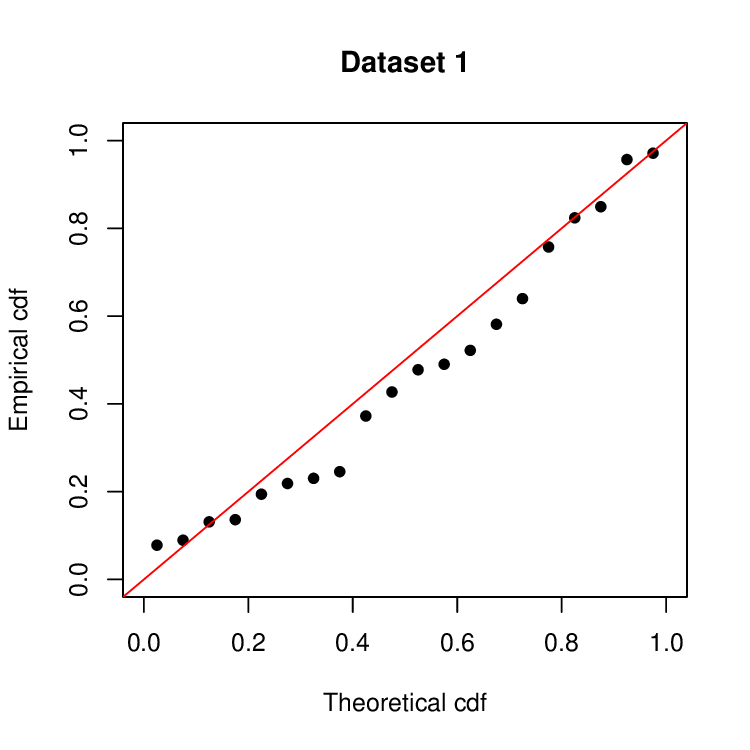}
\includegraphics[width=0.33\textwidth]{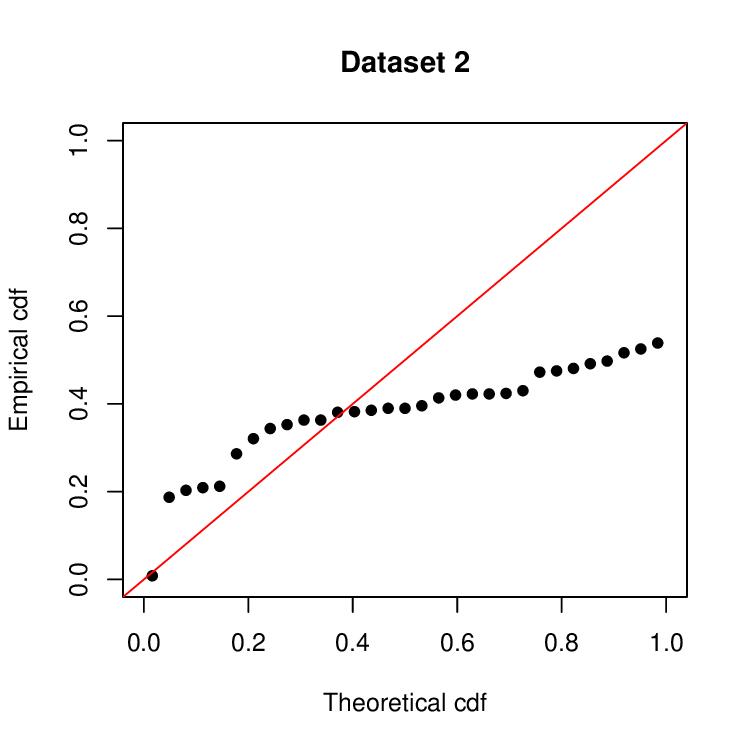}
\end{center}
\caption{The pp-plot between empirical cumulative distribution function and theoretical values of fitted one-sided L\'evy distribution.}
\label{fig:PPplot}
\end{figure}

The quick analysis of the data given by the PP-plot presented in Figure~\ref{fig:PPplot} suggests that the first data set could be effectively modelled by the one-sided L\'evy distribution. However, choosing this distribution for the second data set seems to be incorrect. This observation is confirmed by the results presented in Table~\ref{tab:application}, where the results for all considered test statistics are computed. More specifically, we applied tests $V_n$, $T_n$, $\Delta_n$, $R_{0.2,n}$, and $O_n$ to the considered data sets and calculated the corresponding (simulated) p-values. Note that, for each test, the calculations were based on strong Monte Carlo samples with 100 000 realisations of the underlying test statistic under the null hypothesis.

\begin{table}[ht]
\centering
\begin{tabular}{|r|rrrrr|}
\multicolumn{6}{c}{Dataset 1}\\
  \hline
 & $V_n$ & $T_n$ & $O_n$ & $\Delta_n$ & $R_{0.2,n}$  \\ 
  \hline
TS-value & 1.23 &  -0.04 &-0.01 & 0.08 & -0.35 \\ 
  $p$-value  & 0.56 & 0.91  & 0.49 & 0.90 & 0.82 \\ 
   \hline
\end{tabular}
\begin{tabular}{|r|rrrrr|}
\multicolumn{6}{c}{Dataset 2}\\
  \hline
 & $V_n$ & $\Delta_n$ & $R_{0.2,n}$ & $T_n$ & $O_n$  \\ 
  \hline
TS-value &11.16 & 0.32 & 0.40 & 7.37&  1.96 \\ 
  $p$-value  & $<10^{-4}$ & $<10^{-4}$  & $<10^{-4}$ & $<10^{-4}$  & $<10^{-4}$ \\ 
   \hline
\end{tabular}
\caption{Values of test statistics and the corresponding p-values for the considered data sets.}\label{tab:application}
\end{table}
\FloatBarrier

From the table, we get that there is no evidence to reject the one-sided L\'evy distribution hypothesis for the first data set. However, for the second data set, all three tests agreed to reject the null hypothesis. It should be noted that the results are consistent with the ones presented in~\cite{KumBha2022}. However, note that in this paper, a different normalisation of the test statistic $V_n$ is used which results in a different value of the test statistic. However, the estimated p-values remain unchanged.

\section{Conclusions}
In this paper, we have introduced a novel approach for the scale estimation and goodness-of-fit testing of one-sided L\'evy distribution. Our methodology is based on the quantile conditional moments that are properly defined for the considered distribution in contrast to the unconditional moments (expected value and variance) that are infinite for heavy-tailed distributions. This paper can be considered as the natural extension of our previous research where the methodology based on the quantile conditional variance was proposed for estimation and testing for the general class of \stab distributions. Here, we extended the conditional variance framework to non-symmetric setting and incorporated the sample quantile conditional moments into the scale-ratio statistics in which two estimators of the scale parameter were confronted in the case of one-sided L\'evy distribution. We have shown that the inclusion of the statistics based on quantile conditional moments into the testing scenario made the procedures more effective when confronted with the  approaches based on the ratio of likelihood-based and covariance-based  statistics.  We also studied the probabilistic properties of the introduced test statistics and showed the applications of the proposed methodology using the real data example.

Also, it should be noted that the testing framework based on test statistics $O_n$ and $T_n$ is much more flexible than the testing framework based on test statistic $V_n$. Namely, one can adjust the quantile splits, an input to test statistics $O_n$ and $T_n$, to tailor a test statistic to a specific setting. For example, if we are more interested in testing alternatives in which the tail structure is different, we might use more extreme quantiles in the test statistic definition. That saying, we demonstrated that one could find a generic quantile choice which works reasonably well under various alternatives. Also, we want to emphasize that our framework could be easily expanded to define new goodness-of-fit test statistics based on  (conditional) distribution characteristics, e.g. conditional entropy. This is left for future research.

\section*{Appendix}

In this section we present the consolidated result which aggregated the statements of Proposition~\ref{pr:benchmark}, Proposition~\ref{pr:On.stat}, and Proposition~\ref{pr:Tn.stat}. For brevity, we reformulate all propositions' statements in the consolidated setting and provide joint proof. Before we proceed, let us recall the notation and introduce a few auxiliary objects. As in Section~\ref{S:preliminaries}, for any $n\in \mathbb{N}\setminus \{0\}$, by $(X_1, \ldots, X_n)$ we denote an i.i.d. sample from $Lv(c)$ with some $c>0$. Also, we set $Y_i:=\frac{1}{X_i}$ and $Z_i:=\ln Y_i$, $i=1, \ldots, n$. Next, for a generic split $0<a<b<1$,  the quantile conditional mean $\mu_X(a,b)$, the quantile conditional variance $\sigma^2_{X}$, and their estimators $\hat\mu_X(a,b)$ and $\hat\sigma^2_{X}$, are given by~\eqref{eq:cond_mean_quant},~\eqref{eq:cond_var_quant},~\eqref{eq:estimQCM}, and~\eqref{eq:estimQCV}, respectively. Also, for a generic sequence of random variables $(W_i)_{i=1}^\infty$, we set $\bar{W}_n:=\frac{1}{n}\sum_{i=1}^n W_i$ and $S_{n,Y,Z}:=\frac{1}{n}\sum_{i=1}^n (Y_i-\bar{Y}_n)(Z_i-\bar{Z}_n)$. Finally, we set $\mu_Y:=\mathbb{E}[Y_1]$, $\sigma_{YZ}:=\cov(Y_1,Z_1)$, and we use $Q$ to denote the quantile function of $Lv(1)$ given by~\eqref{eq:Levy_Q}.

\begin{proposition}\label{pr:appendix}
Let $X_i$, $i = 1,2, \ldots$, be i.i.d. random variables following $Lv(c)$ with some $c>0$. Also, let $V_n$, $O_n$, and $T_n$ be given by~\eqref{eq:Vn},~\eqref{eq:On}, and~\eqref{eq:Tn}, respectively. Then, we get
\begin{align*}
    V_n &\overset{d}{\to}\mathcal{N}(0,\tau_V^2), \quad n \to \infty,\\
    O_n&\overset{d}{\to} \mathcal{N}(0, \tau_O^2), \quad n\to\infty,\\
    T_n&\overset{d}{\to} \mathcal{N}(0, \tau_T^2), \quad n\to\infty,
\end{align*}
where $\tau_V > 0$, $\tau_O>0$, and $\tau_T>0$ are constants independent of $c$.
\end{proposition}

\begin{proof}
The idea of the proof is to transform the problem into a multidimensional setting, then use the multivariate central limit theorem and finally apply the delta method. More specifically, we consider

\begin{align}\label{eq.CLT}
I_n:=\sqrt{n}
\begin{bmatrix}
\hat{\mu}_{X}(a_1,b_1) - \mu_X(a_1,b_1) \\
\hat{\mu}_{X}(a_2,b_2) - \mu_X(a_2,b_2) \\
\hat{\sigma}^2_{X} - \sigma_{X}^2 \\
\bar{Y}_n - \mu_{Y} \\
S_{n,Y,Z} - \sigma_{YZ}
\end{bmatrix}
\end{align}
and we show that $I_n$ converges in distribution to some multivariate Gaussian distribution, as $n\to\infty$. We start with providing alternative formulations of the components of $I_n$. Regarding the first component, following the proof of Lemma 2 from~\cite{JelPit2018}, we get
\begin{align}\label{eq:pr:appendix:1}
    \sqrt{n}(\hat{\mu}_X(a_1,b_1) - \mu_X(a_1,b_1)) = \frac{n}{[nb_1] - [na_1]}\frac{1}{n} \sum_{i = 1}^n A_i + r_n,
\end{align}
where $r_n \overset{\mathbb{P}}{\to} 0$ as $n \to \infty$ and, for any $i\in \mathbb{N}\setminus \{0\}$, we set
\begin{align*}
    A_i := (X_i-\mu_X(a_1,b_1))\mathbbm{1}_{\{Q(a_1)<X_i<Q(b_1)\}}+(\mathbbm{1}_{\{X_i < Q(a_1)\}} - a_1)(a_1-\mu_X(a_1,b_1))+(b_1-\mathbbm{1}_{\{X_i<Q(b_1)\}})(b_1-\mu_X(a_1,b_1)).
\end{align*}
Similarly, for the second component, we get
\begin{align}\label{eq:pr:appendix:1.5}
    \sqrt{n}(\hat{\mu}_X(a_2,b_2) - \mu_X(a_2,b_2)) = \frac{n}{[nb_2] - [na_2]}\frac{1}{n} \sum_{i = 1}^n B_i + s_n,
\end{align}
where $s_n \overset{\mathbb{P}}{\to} 0$ as $n \to \infty$ and $B_i$ is a version of $A_i$ with $a_1$ and $b_1$ replaced by $a_2$ and $b_2$, respectively. Next, combining Lemma 3 and Lemma 4 from~\cite{JelPit2018}, we get
\begin{align}\label{eq:pr:appendix:2}
\sqrt{n}(\hat{\sigma}^2_n - \sigma^2_{X}) = \sqrt{n}\frac{n}{[nb] - [na]} \frac{1}{n}\sum_{i =1}^{n} C_i + t_n,
\end{align}
where $t_n \overset{\mathbb{P}}{\to} 0$ as $n \to \infty$ and, for any $i\in \mathbb{N}\setminus \{0\}$,
\begin{multline*}
C_i := \left(\left(X_i - \mu(a,b)\right)^2-\sigma^2(a,b)\right)\mathbbm{1}_{\{Q(a) < X_i < Q(b)\}}+ \left(\mathbbm{1}_{\{X_i \leq Q(a)\}} - a\right)\left( \left(Q(a) - \mu(a,b)\right) - \sigma^2(a,b)\right) \\
+ \left(b-\mathbbm{1}_{\{X_i \leq Q(b) \}}\right)\left(\left(Q(b)-\mu(a,b)\right)^2 - \sigma^2(a,b)\right).
\end{multline*}
Setting $D_i:=Y_i-\mu_Y$, $i\in \mathbb{N}\setminus \{0\}$, we also get
\begin{equation}\label{eq:pr:appendix:3}
    \sqrt{n}\left(\bar{Y}_n - \mu_{Y}\right) = \sqrt{n}\frac{1}{n}\sum_{i=1}^n D_i.
\end{equation}
Finally, by direct computation, we get
\begin{align}\label{eq:pr:appendix:4}
\sqrt{n} \left(S_{n,Y,Z} - \sigma_{YZ}\right) = \sqrt{n}\frac{1}{n} \sum_{i =1}^n E_i+v_n
\end{align}
where, for any $i\in \mathbb{N}\setminus \{0\}$, we set
\begin{align*}
    E_i&:=(Y_i-\mu_Y)(Z_i-\mu_Z) - \sigma_{YZ},\\
    v_n&:=-\sqrt{n}(\bar{Y}_n-\mu_Y)(\bar{Z}_n-\mu_Z).
\end{align*}
By the law of large numbers, we get $\bar{Z}_n\overset{\mathbb{P}}{\to}\mu_Z$, while from the central limit theorem, we get that $\sqrt{n}(\bar{Y}_n-\mu_Y)$ converges in law to a non-degenerate Gaussian distribution as $n\to\infty$. These observations combined with Slutsky's theorem imply that $v_n\overset{\mathbb{P}}{\to}0$, as $n\to\infty$.

Now, combining~\eqref{eq:pr:appendix:1}--\eqref{eq:pr:appendix:4}, we get
\begin{align}
I_n &=M_n
\sqrt{n} \frac{1}{n} \sum_{i = 1}^n
F_i
+  g_n, \quad n\in \mathbb{N},
\end{align}
where, for any $n\in \mathbb{N}$, we set $F_n:=[A_n, B_n, C_n, D_n, E_n]^T$, $g_n:=[r_n, s_n, t_n, 0, v_n]^T$, and $M_n$  is $4\times 4$ diagonal matrix with the main diagonal given by $[\frac{n}{[nb] - [na]},\frac{n}{[nb] - [na]},0,0,0]$. In particular, we get that $(F_n)$ is a family of i.i.d. square integrable random vectors with zero expectation and $F_n\overset{\mathbb{P}}{\to} 0$, as $n\to\infty$. Thus, noting that $M_n$ converges to the diagonal matrix $ M$ with the main diagonal $[\frac{1}{b-a},\frac{1}{b-a},0,0,0]$ and using the multivariate central limit theorem combined with Slutsky's theorem, we get that
$I_n \overset{d}{\to}\mathcal{N}(\begin{bmatrix} 0, 0, 0,0,0\end{bmatrix}^T, \Sigma)$, as $n\to\infty$, where 
\begin{align*}
\Sigma := M
\begin{bmatrix}
D^2(A_1) & \cov(A_1, B_1) & \cov(A_1,C_1) & \cov(A_1,D_1) & \cov(A_1,E_1) \\
\cov(B_1, A_1) & D^2(B_1) & \cov(B_1,C_1) & \cov(B_1,D_1) & \cov(B_1,E_1)\\
\cov(C_1,A_1) & \cov(C_1,B_1) & D^2(C_1) & \cov(C_1,D_1) & \cov(C_1,E_1) \\
\cov(A_1,D_1) & \cov(B_1,D_1) & \cov(C_1,D_1) & D^2(D_1) & \cov(D_1,E_1)\\
\cov(A_1,E_1) & \cov(B_1,E_1) & \cov(C_1,E_1) & \cov(D_1,E_1) &  D^2(E_1)
\end{bmatrix} M.
\end{align*}

Next, let us define $H_n:=[\hat\mu_X(a_1,b_1), \hat\mu_X(a_2,b_2), \hat\sigma^2_X, \bar{Y}_n, S_{n,Y,Z}]^T$, $H:=[\mu_X(a_1,b_1), \mu_X(a_2,b_2), \sigma^2_X, \mu_Y, \sigma_{Y,Z}]^T$, and note that $I_n=\sqrt{n}(H_n-H)\overset{d}{\to}\mathcal{N}([0,0,0,0,0]^T,\Sigma)$. In this setting, we get
\begin{align*}
    V_n&=\sqrt{n}(g_V(H_n)-g_V(H)), \quad \text{where} \quad g_V(x,y,z,w,v) = \frac{2w}{v},\\
    O_n&=\sqrt{n}(g_O(H_n)-g_O(H)), \hat{\sigma}^2(a,b)_X, \bar{Y},S_{n,Y,Z}), \quad \text{where} \quad g_O(x,y,z,w,v) = \frac{x \cdot \mu_{Lv(1)}(a_2,b_2)}{y \cdot \mu_{Lv(1)}(a_1,b_1)},\\
    T_n&=\sqrt{n}(g_T(H_n)-g_T(H)), \quad \text{where} \quad g_T(x,y,z,w,v)=\frac{w}{v} + \frac{x w}{2\mu_{Lv(1)}(a_1,b_1)}.
\end{align*}
Thus, using the delta method (see e.g. Theorem 7 in~\cite{Ferguson1996}), we get
\begin{align*}
V_n &\overset{d}{\to}\mathcal{N}(0,(\nabla g_V(H))^T \Sigma \nabla g_V(H)), \quad n \to\infty,\\
O_n &\overset{d}{\to}\mathcal{N}(0,(\nabla g_O(H))^T \Sigma \nabla g_O(H)), \quad n \to\infty,\\
T_n &\overset{d}{\to}\mathcal{N}(0,(\nabla g_T(H))^T \Sigma \nabla g_T(H)), \quad n \to\infty,
\end{align*}
which concludes the proof. 
\end{proof}


\section*{Acknowledgements}
Marcin Pitera and Agnieszka Wy\l{}oma\'{n}ska acknowledge support from the National Science Centre, Poland, via project 2020/37/B/HS4/00120. Part of the work of Damian Jelito was funded by the Priority Research Area Digiworld under the program Excellence Initiative – Research University at the Jagiellonian University in Krak\'{o}w.

\begin{footnotesize}
\bibliographystyle{abbrvnat}
\bibliography{mybibliography}
\end{footnotesize}


\end{document}